\PassOptionsToPackage{svgnames}{xcolor}
\documentclass{lmcs} 
\pdfoutput=1
\usepackage[utf8]{inputenc}

\usepackage{lastpage}
\lmcsdoi{16}{1}{32}
\lmcsheading{}{\pageref{LastPage}}{}{}%
{Dec.~21,~2017}{Mar.~23,~2020}{}

\pdfoutput=1


\usepackage{tikz-cd}
\usepackage{etoolbox} \usepackage{wrapfig} \usepackage{tikz}
\usetikzlibrary{decorations.markings} \tikzset{negated/.style={
    decoration={markings, mark= at position 0.5 with { \node[transform
        shape,xscale=.8,yscale=.4] (tempnode) {$\slash$}; } },
    postaction={decorate} } } \usepackage{amsmath}
\usepackage{scalerel} \usepackage{amssymb} \usepackage{mathtools}

\usepackage[hidelinks]{hyperref}
\usepackage[paper]{knowledge}
\knowledgeconfigure{notion,quotation,protect quotation={tikzcd}}

\theoremstyle{plain} 

\usepackage{cite}

\knowledgedirective{math notion}{notion}

\knowledge {category representing input words}{}

\knowledge {dfaclassic}{notion}

\knowledge {statesclassic}{notion} \knowledge {initialclassic}{notion}
\knowledge {finalclassic}{notion} \knowledge {transition
  classic}{notion} \knowledge {runclassic}{notion} \knowledge
{acceptingclassic}{notion} \knowledge {acceptsclassic}{synonym}
\knowledge {languageclassic}{notion}

\knowledge {non-deterministic
  automaton}{notion}
\knowledge {non-deterministic
  automata}{synonym}

\knowledge{acceptingelements}{notion}

\knowledge {$\catC $-automaton}{notion} \knowledge
{$\catC $-automata}{synonym} \knowledge {$\catD $-automaton}{synonym}
\knowledge {automata in the category}{synonym} \knowledge {automaton
  in $\KlTrans $}{synonym} \knowledge {automaton in $\Set $}{synonym}
\knowledge {automata}{synonym} \knowledge {automata in the Kleisli
  category of a suitably chosen monad}{synonym} \knowledge
{automaton}{synonym}

\knowledge {$\catC $-language}{notion} \knowledge
{$\catC $-languages}{synonym} \knowledge {languages}{synonym}
\knowledge {language}{synonym}

\knowledge {\Auto }{notion}

\knowledge {accept}[accepts]{notion} \knowledge {$\KlTrans $-automata
  for $\langLKlT $}{synonym} \knowledge {$\Set $-automata accepting
  $\langLSet $}{synonym} \knowledge {for a language}{synonym}
\knowledge {language accepted}{synonym} \knowledge {automata for a
  language}{synonym} \knowledge {automata for languages}{synonym}
\knowledge {accepting}{synonym} \knowledge{$(\Set ,1,2)$-automata}{synonym}
\knowledge{$(\Rel ,1,1)$-automata}{synonym}
\knowledge{$\Set $-automata accepting $\langLSetTrans $}{synonym}

\knowledge {initial automaton}{notion} \knowledge {initial
  $\langL $-automaton}{synonym}

\knowledge {final automaton}{notion} \knowledge {final
  automata}{synonym} \knowledge {final $\langL $-automaton}{synonym}

\knowledge {\defcatIword }{math notion} \knowledge {\defcatOword }{math
  notion} \knowledge {\defcatIMon }{math notion} \knowledge {\defcatOMon
}{math notion}

\knowledge {\catIword }{math notion} \knowledge {\catOword }{math
  notion} \knowledge {\catIMon }{math notion} \knowledge {\catOMon
}{math notion}

\knowledge {\objectLeft }{math notion} \knowledge {\objectRight }{math
  notion} \knowledge {\symbLeft }{math notion} \knowledge {\symbRight
}{math notion} \knowledge {\objectCenter }{math notion}

\knowledge {\defobjectLeft }{math notion} \knowledge {\defobjectRight }{math
  notion} \knowledge {\defsymbLeft }{math notion} \knowledge {\defsymbRight
}{math notion} \knowledge {\defobjectCenter }{math notion}

\knowledge {$(\catC , X,Y)$-automaton}{notion} \knowledge
{$(\KlTrans ,1,1)$-automaton}{synonym} \knowledge
{$(\KlTrans ,1,1)$-automata}{synonym} \knowledge
{$(\Set ,1,\alphabetB ^*)$-automata}{synonym}

\knowledge {$(\catC , X,Y)$-language}{notion} \knowledge
{$(\catC ,I,GO)$-languages}{synonym} \knowledge
{$(\catD ,FI,O)$-languages}{synonym} \knowledge
{$(\catC ,I,GO)$-language}{synonym} \knowledge
{$(\catD ,FI,O)$-language}{synonym} \knowledge
{$(\Set ,1,2)$-language}{synonym} \knowledge
{$(\Setop ,2,1)$-language}{synonym} \knowledge
{$(\Rel ,1,1)$-language}{synonym} \knowledge
{$(\KlT , 1,1)$-language}{synonym} \knowledge
{$(\Set , 1,B^*+1)$-language}{synonym} \knowledge
{$(\Set ,1,\alphabetB ^*+1)$-automata}{synonym} \knowledge
{$(\KlT , \FreeTrans 1,1)$-language}{synonym} \knowledge
{$(\KlTrans ,1,1)$-language}{synonym} \knowledge
{$(\KlTrans ,1,1)$-languages}{synonym} \knowledge
{$(\Set , 1,\UTrans 1)$-language}{synonym}

\knowledge{math notion}
 | \langLSet 
\knowledge{math notion}
 | \langLSetop 
\knowledge{math notion}
 | \langLRel 
\knowledge{math notion}
 | \langLSetTrans 
\knowledge{math notion}
 | \langLEMTTrans

\knowledge\catReach{math notion} \knowledge\catObs{math notion}

\knowledge {\Min }{math notion} \knowledge {\Obs }{math notion}
\knowledge {\Reach }{math notion}

\knowledge {-divides}{notion} \knowledge {-divisibility}{synonym}

\knowledge {-quotient}{notion} \knowledge {-quotients}{synonym}

\knowledge {-subobject}{notion} \knowledge {-subobjects}{synonym}

\knowledge {-minimal}{notion} \knowledge {minimal automaton}{synonym}
\knowledge {minimal}{synonym}

\knowledge \EpiAut{math notion} \knowledge \MonoAut{math notion}

\knowledge {\EpiAutL }{math notion} \knowledge {\MonoAutL }{math
  notion}

\knowledge {observable}{notion} \knowledge {observable quotient
  automaton}{synonym}

\knowledge {reachable}{notion} \knowledge {reachable
  sub-automaton}{synonym}

\knowledge {subsequential transducer}{notion} \knowledge
{Subsequential transducers}{synonym} \knowledge {subsequential
  transducers}{synonym}

\knowledge {input alphabetSS}{notion} \knowledge {output
  alphabetSS}{notion} \knowledge {statesSS}{notion} \knowledge
{initial state of the transducer}{notion} \knowledge {initial states
  of the transducers}{synonym} \knowledge {initialization
  value}{notion}

\knowledge {initialSS}{notion}

\knowledge {\semTrans }{math notion}

\knowledge {partial transition function for the letter~$a$}{notion}
\knowledge {partial transition function}{synonym}

\knowledge {partial termination function}{notion}

\knowledge {partial production function for the letter~$a$}{notion}
\knowledge {partial production function}{synonym}

\knowledge {initial transducer}{notion} \knowledge {final
  transducer}{notion} \knowledge {factorization transducer}{notion}

\knowledge {\monadTrans}{math notion} \knowledge {\KlTrans}{math
  notion} \knowledge {\EMT}{math notion}

\knowledge {Kleisli category}{notion} \knowledge {Kleisli
  categories}{synonym} \knowledge {Kleisli category for
  $\monadTrans $}{synonym}

\knowledge {Kleisli morphism}{notion} \knowledge {Kleisli
  morphisms}{synonym}

\knowledge {\FreeTrans}{math notion}
\knowledge {\UTrans}{math notion}
\knowledge {\UTransEM}{math notion}
\knowledge {\FreeTransEM}{math notion}

\knowledge{\bijFinalTrans}{math notion}

\knowledge {\botfun }{math notion}

\knowledge {longest common prefix}{notion} \knowledge {\lcp }{math
  notion}

\knowledge {irreducible function}{notion} \knowledge
{irreducible}{synonym} \knowledge {\Irr }{math notion}

\knowledge {reduction of~$K$}{notion} \knowledge {\red }{math notion}

\knowledge {\suff }{math notion} \knowledge {\pstar}{math notion}

\knowledge {Brzozowski's algorithm}{notion} \knowledge {Brzozowski's
  minimization algorithm}{synonym}

\knowledge {\FPow }{math notion} \knowledge {\UPow }{math notion}
\knowledge {\UPowop }{math notion} \knowledge {\FPowop }{math notion}

\knowledge {\EpiKlT }{math notion} \knowledge {\MonoKlT }{math notion}

\knowledge {\determinize }{math notion} \knowledge {\transpose }{math
  notion} \knowledge {\codeterminize }{math notion} \knowledge
{determinization functor}{notion} \knowledge {codeterminization
  functor}{notion}

\knowledge {category}{} \knowledge {isomorphic}{} \knowledge
{object}{} \knowledge {objects}{synonym}

\knowledge {initial object}{} \knowledge {initial}{synonym} \knowledge
{final object}{} \knowledge {final}{synonym}

\knowledge {isomorphism}{} \knowledge {factorization}{} \knowledge
{factorization systems}{} \knowledge {factorization system}{}

\knowledge {adjunctions}{} \knowledge {-factorization}{} \knowledge
{diagonal property}{} \knowledge {Kan extensions}{}

\knowledge {regular}{} \knowledge {regular epis}{} \knowledge
{algebraic categories}{} \knowledge {quotients of algebras}{}
\knowledge {subalgebras}{} \knowledge {composition}{}

\knowledge {monoid recognizer}[monoid recognizers]{notion} \knowledge
{surjective monoid recognizer}[surjective monoid recognizers]{notion}
\knowledge {$A^*$-biaction}{notion} \knowledge
{$A^*$-biactions}{synonym} \knowledge {$A^*$-biaction
  recognizer}{notion} \knowledge {$A^*$-biaction recognizers}{synonym}
\knowledge {biaction recognizers}{synonym} \knowledge {surjective
  $A^*$-biaction recognizer}{notion}

\knowledge {\wmorph }{math notion} \knowledge {\wfinmorph }{math
  notion} \knowledge {\uvmorph }{math notion} \knowledge {\uvfinmorph
}{math notion} \knowledge {$(\catIMon ,\Set ,1,2)$-automaton}{notion}
\knowledge {$(\catIMon ,\Set ,1,2)$-automata}{synonym}

\knowledge {\Syn }{math notion}

\newrobustcmd\defcatIword{\kl[\defcatIword]{\mathcal{I}_{\mathsf{word}}}}
\newrobustcmd\defcatOword{\kl[\defcatOword]{\mathcal{O}_{\mathsf{word}}}}

\newrobustcmd\cat[1]{\mathcal{#1}} \newrobustcmd\catA{\cat A}
\newrobustcmd\catC{\cat C} \newrobustcmd\catD{\cat D}
\newrobustcmd\catE{\cat E} \newrobustcmd\catI{\cat I}
\newrobustcmd\catIMon{\kl[\catIMon]{\mathcal{I}_{\mathsf{Mon}}}}
\newrobustcmd\catOMon{\kl[\catOMon]{\mathcal{O}_{\mathsf{Mon}}}}
\newrobustcmd\catIword{\kl[\catIword]{\mathcal{I}_{\mathsf{word}}}}
\newrobustcmd\catOword{\kl[\catOword]{\mathcal{O}_{\mathsf{word}}}}
\newrobustcmd\catO{\cat O} \newrobustcmd\catS{\cat S}
\newrobustcmd\catK{\cat K}

\newrobustcmd\op{\mathit{op}}
\newrobustcmd\Auto{\kl[\Auto]{\mathsf{Auto}}}
\newrobustcmd\catAutoL{\Auto(\langL)}
\newrobustcmd\catAutoLC{\Auto(\langL_\catC)}
\newrobustcmd\catAutoLD{\Auto(\langL_\catD)}
\newrobustcmd\catAutoLSet{\texorpdfstring{\Auto(\langLSet)}{Auto(LSet)}}
\newrobustcmd\catAutoLSMod{\Auto(\langL_{\SMod})}
\newrobustcmd\catAutoLSetop{\Auto(\langLSetop)}
\newrobustcmd\catAutoLSModop{\Auto(\langL_{\SMod^{\mathit{op}}})}
\newrobustcmd\catAutoLRel{\Auto(\langL_\Rel)}
\newrobustcmd\catAutoLRelop{\Auto(\langL_\Relop)}
\newrobustcmd\catAutoLSetrev{\Auto(\langL_\Set^\inv)}
\newrobustcmd\Set{\mathsf{Set}}
\newrobustcmd\SMod{S\textrm{-}\mathsf{Mod}}
\newrobustcmd\KMod{K\textrm{-}\mathsf{Mod}}
\newrobustcmd\Setop{\mathsf{Set^\op}} \newrobustcmd\Rel{\mathsf{Rel}}
\newrobustcmd\Relop{\mathsf{Rel^\op}}

\newrobustcmd\lang[1]{\mathcal{#1}}
\newrobustcmd\langL{\lang L}
\newrobustcmd\langLSet{\kl[\langLSet]{\lang L_{\Set}}}
\newrobustcmd\langLSetop{\kl[\langLSetop]{\lang L_{\Setop}}}
\newrobustcmd\langLRel{\kl[\langLRel]{\lang L_{\Rel}}}
\newrobustcmd\langLRelop{\kl[\langLRelop]{\lang L_{\Relop}}}
\newrobustcmd\langLRelrev{\kl[\langLRelrev]{\lang L_{\Rel}^{\inv}}}
\newrobustcmd\langLSetrev{\kl[\langLSetrev]{\lang L_{\Set}^{\inv}}}
\newrobustcmd\aut[1]{\mathcal{#1}}
\newrobustcmd\biactAut[1]{\mathcal{#1}}
\newrobustcmd\autA{\aut A}
\newrobustcmd\autB{\aut B}
\newrobustcmd\autC{\aut C}
\newrobustcmd\biactAutA{\biactAut{A}}
\newrobustcmd\biactAutB{\biactAut{B}}

\newrobustcmd\Reach{\kl[\Reach]{\mathtt{Reach}}}
\newrobustcmd\Obs{\kl[\Obs]{\mathtt{Obs}}}
\newrobustcmd\Min{\kl[\Min]{\mathtt{Min}}}
\newrobustcmd\Syn{\kl[\Syn]{\mathtt{Syn}}}
\newrobustcmd{\ReachA}{\Reach{(\autA)}}
\newrobustcmd{\ObsA}{\Obs{(\autA)}} \newrobustcmd\MinL{\Min{(\langL)}}

\newrobustcmd\AutInit{\autA^{\kl[initial automaton]{\mathit{init}}}}
\newrobustcmd\AutFin{\autA^{\kl[final automaton]{\mathit{final}}}}

\newrobustcmd\autAinitL{\AutInit(\langL)}
\newrobustcmd\autAfinalL{\AutFin(\langL)}

\newrobustcmd\defsymbLeft{\kl[\defsymbLeft]\triangleright}
\newrobustcmd\defsymbRight{\kl[\defsymbRight]\triangleleft}
\newrobustcmd\defobjectLeft{\kl[\defobjectLeft]{\mathsf{in}}}
\newrobustcmd\defobjectRight{\kl[\defobjectRight]{\mathsf{out}}}
\newrobustcmd\defobjectCenter{\kl[\defobjectCenter]{\mathsf{states}}}

\newrobustcmd\symbLeft{\kl[\symbLeft]\triangleright}
\newrobustcmd\symbRight{\kl[\symbRight]\triangleleft}
\newrobustcmd\objectLeft{\kl[\objectLeft]{\mathsf{in}}}
\newrobustcmd\objectRight{\kl[\objectRight]{\mathsf{out}}}
\newrobustcmd\objectCenter{\kl[\objectCenter]{\mathsf{states}}}

\newrobustcmd\wmorph[1]{\kl[\wmorph]{\mathtt{#1}}}
\newrobustcmd\wfinmorph[1]{\kl[\wfinmorph]{\mathtt{\overline{{#1}}}}}

\newrobustcmd\msquare{\mathord{\scalerel*{\square}{o}}}
\newrobustcmd\uvmorph[2]{\kl[\uvmorph]{\mathtt{#1}\msquare
    \mathtt{#2}}}
\newrobustcmd\uvfinmorph[2]{\kl[\uvfinmorph]{\overline{\mathtt{#1}\msquare\mathtt{#2}}}}

\newrobustcmd\EpiAutL{\kl[\EpiAutL]{\mathcal{E}_{\mathsf{Auto}(\langL)}}}
\newrobustcmd\MonoAutL{\kl[\MonoAutL]{\mathcal{M}_{\mathsf{Auto}(\langL)}}}

\newrobustcmd\EpiAut{\kl[\EpiAutL]{\mathcal{E}_{\mathsf{Auto}(\langL_{\KlTrans})}}}
\newrobustcmd\MonoAut{\kl[\MonoAutL]{\mathcal{M}_{\mathsf{Auto}(\langL_{\KlTrans})}}}

\newrobustcmd\Epi{\ensuremath{\mathcal{E}}}
\newrobustcmd\Mono{\ensuremath{\mathcal{M}}}

\newrobustcmd\Lan[2]{\mathsf{Lan}_{#2}{#1}}
\newrobustcmd\Ran[2]{\mathsf{Ran}_{#2}{#1}}

\newrobustcmd\St{\mathsf{State}} \newrobustcmd\MB{M(B)}
\newrobustcmd\tokl\nrightarrow \newrobustcmd\TM{T_{M}}

\newrobustcmd\catAutoLKlT{\texorpdfstring{\Auto(\langL_{\KlTrans})}{Auto(LKl(T))}}
\newrobustcmd\catAutoLEMT{\Auto(\langLEMTTrans)}
\newrobustcmd\monadTrans{\kl[\monadTrans]{\mathcal{T}}}
\newrobustcmd\KlTrans{\texorpdfstring{\kl[\KlTrans]{\mathsf{Kl}(\mathcal{T})}}{Kl(T)}}
\newrobustcmd\EMT{\kl[\EMT]{\mathsf{EM}(\mathcal{T})}}
\newrobustcmd\langLKlT{\texorpdfstring{\lang L_{\KlTrans}}{LKl(T)}}
\newrobustcmd\EpiKlT{\kl[\EpiKlT]{\mathcal{E}_{\mathsf{Kl}(\mathcal{T})}}}
\newrobustcmd\MonoKlT{\kl[\MonoKlT]{\mathcal{M}_{\mathsf{Kl}(\mathcal{T})}}}
\newrobustcmd\semTrans[1]{\kl[\semTrans]{[\![}#1\kl[\semTrans]{]\!]}}
\newrobustcmd\KlT{\KlTrans}
\newrobustcmd\FreeTrans{\kl[\FreeTrans]{F_{\mathcal{T}}}}
\newrobustcmd\FreeTransEM{\kl[\FreeTransEM]{F^{\mathcal{T}}}}
\newrobustcmd\UTrans{\kl[\UTrans]{U_{\mathcal{T}}}}
\newrobustcmd\UTransEM{\kl[\UTransEM]{U^{\mathcal{T}}}}

\newrobustcmd\langLSetTrans{\kl[\langLSetTrans]{\lang L_{\Set}}}
\newrobustcmd\langLEMTTrans{\kl[\langLEMTTrans]{\lang L_{\EMT}}}

\newrobustcmd\pstar{\mathrel{\kl[\pstar]{\star}}}

\newrobustcmd\bijFinalTrans{\kl[\bijFinalTrans]{\varphi}}

\newrobustcmd\FreeAut{\overline{\FreeTrans}}
\newrobustcmd\FreeAutEM{\overline{\FreeTransEM}}
\newrobustcmd\AutUTrans{\overline{\UTrans}}
\newrobustcmd\AutUTransEM{\overline{\UTransEM}}
\newrobustcmd\dom{\mathrm{dom}}
\newrobustcmd\FPow{\kl[\FPow]{F_{\mathcal{P}}}}
\newrobustcmd\UPow{\kl[\UPow]{U_{\mathcal{P}}}}
\newrobustcmd\Pow{\mathcal{P}} \newrobustcmd\rev{\mathcal{R}}
\newrobustcmd\FPowop{\kl[\FPowop]{F^\op_{\mathcal{P}}}}
\newrobustcmd\UPowop{\kl[\UPowop]{U^\op_{\mathcal{P}}}}

\newrobustcmd\botfun{\kl[\botfun]{\kappa_\bot}}
\newrobustcmd\IrrAB{\kl[\Irr]{\mathsf{Irr}}(A^*,B^*)}
\newrobustcmd\red{\kl[\red]{\mathsf{red}}}
\newrobustcmd\lcp{\kl[\lcp]{\mathsf{lcp}}} \newrobustcmd\redL{\red(L)}
\newrobustcmd\redK{\red(K)} \newrobustcmd\lcpK{\lcp(K)}
\newrobustcmd\lcpL{\lcp(L)}

\newrobustcmd\suff{\kl[\suff]{\mathsf{suff}}}
\newrobustcmd\inv{{\mathsf{rev}}}

\newrobustcmd\swap{\mathsf{swap}}

\renewrobustcmd\Im{\mathrm{Im}}

\newrobustcmd\catObsAutoLSetop{\kl[\catObs]{\mathsf{Obs}(\langL_{\Setop})}}
\newrobustcmd\catObsAutoLSModop{\kl[\catObs]{\mathsf{Obs}(\langL_{\SMod^{\mathit{op}}})}}
\newrobustcmd\catReachAutoLSet{\kl[\catReach]{\mathsf{Reach}(\langL_{\Set})}}
\newrobustcmd\catReachAutoLSMod{\kl[\catReach]{\mathsf{Reach}(\langL_{\SMod})}}
\newrobustcmd\catObsAutoLC{\kl[\catObs]{\mathsf{Obs}(\langL)}}
\newrobustcmd\catReachAutoLC{\kl[\catReach]{\mathsf{Reach}(\langL)}}
\newrobustcmd\catReachK{\kl[\catReach]{\mathsf{Reach}(\catK)}}
\newrobustcmd\catReachKop{\kl[\catReach]{\mathsf{Reach}(\catK^\mathit{op})}}
\newrobustcmd\catObsK{\kl[\catObs]{\mathsf{Obs}(\catK)}}
\newrobustcmd\catObsKop{\kl[\catObs]{\mathsf{Obs}(\catK^\mathit{op})}}

\newrobustcmd\determinize{\kl[\determinize]{\mathtt{determinize}}}
\newrobustcmd\codeterminize{\kl[\codeterminize]{\mathtt{codeterminize}}}
\newrobustcmd\transpose{\kl[\transpose]{\mathtt{transpose}}}

\begin{document}

\title[Automata Minimization: a Functorial Approach]{Automata
  Minimization: a Functorial Approach\rsuper*}
\titlecomment{{\lsuper*}This is the journal version
  of~\cite{ColcombetPetrisan:CALCO2017}}

\author[T.~Colcombet]{Thomas Colcombet} 
\address{CNRS, IRIF, Université de Paris, France} 
\email{\{thomas.colcombet,petrisan,\}@irif.fr} 

\author[D.~Petri{\c s}an]{Daniela Petri{\c s}an} 
\thanks{This work was supported by the European Research Council (ERC)
  under the European Union’s Horizon 2020 research and innovation
  programme (grant agreement No.670624), and by the DeLTA ANR project
  (ANR-16-CE40-0007). The authors also thank the Simons Institute for
  the Theory of Computing where this work has been partly
  developed.} 





\begin{abstract}
  In this paper we regard languages and their acceptors -- such as
  deterministic or weighted automata, transducers, or monoids -- as
  functors from input categories that specify the type of the
  languages and of the machines to categories that specify the type of
  outputs.

  Our results are as follows:
  a) We provide sufficient conditions on the output category so that
  minimization of the corresponding automata is guaranteed.
  b) We show how to lift adjunctions between the categories for output
  values to adjunctions between categories of automata.
  c) We show how this framework can be instantiated to unify several
  phenomena in automata theory, starting with determinization,
  minimization and syntactic algebras.  We provide explanations of
  Choffrut's minimization algorithm for subsequential transducers and
  of Brzozowski's minimization algorithm in this setting.
\end{abstract}

\maketitle

\section{Introduction}
\label{sec:intro}

There is a long tradition of interpreting results of automata theory
through the lens of category theory.  Typical instances of this scheme
interpret automata as algebras (together with a final map) as put
forward in~\cite{ArbibManes75,goguen1972,AdamekTrnkova89}, or as
coalgebras (together with an initial map), see for
example~\cite{Jacobs97atutorial,RUTTEN20003}. This dual narrative
proved very useful~\cite{BonchiBHPRS14} in explaining at an abstract
level "Brzozowski's minimization algorithm" and the duality between
reachability and observability (which goes back all the way to the
work of Arbib and Manes~\cite{ArbibManes75} and Kalman~\cite{Kalman}).

In this paper, we adopt a slightly different approach, and we define
directly the notion of an "automaton" (over finite words) as a functor
from a "category representing input words", to a category representing
the computation and output spaces. For example, deterministic automata
are represented as functors valued in the category of sets and
functions, non-deterministic automata as functors valued in the
category of sets and relations, while weighted automata over a
semiring $S$ as functors valued in the category of $S$-modules. The
notions of a "language" and of "language accepted" by an automaton are
adapted along the same pattern.

We provide several developments around this idea. First, we recall
(see~\cite{DBLP:conf/mfcs/ColcombetP17}) that the existence of a
"minimal automaton" for a "language" is guaranteed by the existence of
an "initial" and a "final automaton" in combination with a
"factorization system". The idea of using factorization systems in the
context of minimization has of course a long history, going back at
least to Goguen~\cite{goguen1972}. However, the functorial presentation that we
adopted allows us to give a unifying perspective of the minimization
of various forms of automata and algebraic structures used for
language recognition. Additionally, we explain how, in the functor
presentation that we have adopted, the existence of "initial" and
"final automata" "for a language" can be phrased in terms of "Kan
extensions". As an immediate corollary, we identify sufficient
conditions on the output category for the existence of the
corresponding minimal automaton: existence of certain limits and
colimits, as well as of a suitable factorization system.

We also show how adjunctions between categories can be lifted to the
level of "automata for languages" in these categories
(Lemma~\ref{lem:lifting-adjunctions}). This lifting accounts for
several constructions in automata theory, determinization to start
with.  Indeed, determinization of automata can be understood via a
lifting of the Kleisli adjunction between the categories $\Rel$ (of
sets and relations) and $\Set$ (of sets and functions); and reversing
non-deterministic automata can be understood via a lifting of the
self-duality of $\Rel$.

We then use this framework in order to explain several well-known
constructions in automata theory.

The most involved contribution (Theorem~\ref{theorem:minimization
  transducer}) is to rephrase in this framework the minimization
result for "subsequential transducers" due to
Choffrut~\cite{Choffrut79}. We do this by instantiating the category
of outputs with the "Kleisli category" for the monad
$\monadTrans X=B^*\times X+1$, where $B$ is the output alphabet of the
transducers. In this case, despite the lack of completeness of the
ambient category, one can still prove the existence of an initial and
of a final automaton, as well as, surprisingly, of a factorization
system.

The second concrete application presented in Section~\ref{sec:brzozowski}
is a proof of correctness of Brzozowski's minimization algorithm, for
both deterministic and weighted
automata.  
Brzozowski's minimization algorithm for deterministic automata can be
understood by lifting the adjunctions between $\Set$ and its opposite
category $\Set^\op$, as an immediate application of
Lemma~\ref{lem:lifting-adjunctions}. Similarly, upon viewing weighted
automata as functors valued in the category $\SMod$ of $S$-modules, a
weighted version of Brzozowski's minimization algorithm described
in~\cite{BonchiBHPRS14} can be explained by lifting the adjunction
between $\SMod$ and its opposite $\SMod^\op$.

Lastly, in Section~\ref{sec:syntactic-monoids} we show how the
syntactic monoid for a language can be obtained in the same spirit as
the minimal automaton. To this end, we replace the category
representing finite words with one suitable for representing biaction
and monoid recognizers of languages.

\subparagraph{\textbf{Related work.}}

Many of the constructions outlined here have already been explained
from a category-theoretic perspective, using various techniques. For
example, the relationship between minimization and duality was subject
to numerous papers, see for
example~\cite{Bezhanishvili2012,BonchiBHPRS14, BonchiBRS12} and the
references therein.  The coalgebraic perspective on minimization was
also emphasized in papers such
as~\cite{AdamekMMS:Well-pointed-coalg,AdamekBHKMS:2012,SilvaEtAl:genPow}. We
briefly mention the relationship with well-pointed coalgebras in
Remark~\ref{rem:comp-well-pointed-coalgebras}. However, we argue that
in some instances the functorial approach may be better suited for
explaining minimization, an example being that of subsequential
transducers, see Section~\ref{sec:choffrut}.

In~\cite{Hansen10} subsequential structures (i.e. subsequential
transducers without initial state and inital prefix) are modeled as
coalgebras for an endofunctor on $\Set$. However, the corresponding
notion of coalgebra morphism does not accurately capture the suitable
notion of subsequential morphisms. In Section~\ref{sec:choffrut} we
model subsequential transducers as functors valued in a "Kleisli
category" $\KlT$. This category does not have powers, hence working
with coalgebras for an endofunctor on $\KlT$ is not possible, see
Remark~\ref{rem:comp-well-pointed-coalgebras}. On the other hand, if
one uses instead coalgebras for a $\Set$-endofunctor, as
in~\cite{Hansen10}, then only certain ``normalized'' subsequential
structures can be fully dealt with coalgebraically.

Understanding determinization and codeterminization by lifting
adjunctions to coalgebras was considered in~\cite{KerstanKW14}, and is
related to our results from Section~\ref{sec:det-codet-adj}.

The paper which is closest in spirit to our work is a seemingly
forgotten paper~\cite{Bainbridge74}. However, in this work, Bainbridge
models the \emph{state space} of the machines as a functor.  Left and
right Kan extensions are featured in connection with the initial and
final automata, but in a slightly different
setting. Lemma~\ref{lem:lifting-adjunctions}, which albeit technically
simple, has surprisingly many applications, builds directly on his
work.

The functorial approach to non-deterministic automata presented in this
paper is reminiscent of the work on automata in quantaloid-enriched
categories developed in~\cite{Rosenthal95}. However, in this paper we
do not consider neither relational presheaves (which are \emph{lax}
functors) nor enriched categories. We also aimed to keep the
requirements on the output category as simple as possible.

\newrobustcmd\alphabetA{A} \newrobustcmd\alphabetB{B}

\section{Languages and Automata as Functors}
\label{sec:lang-auto-functors}

In this section, we introduce the notion of automata via functors, and
this is the common denominator of the different contributions of this
paper. We then discuss automata minimization in this generic setting.

\subsection{Automata as functors}
We introduce automata as functors starting from the special case of
classical "deterministic automata@dfaclassic".  \AP In the standard
definition, a ""deterministic automaton@dfaclassic"" is a tuple:
\begin{align*}
  \langle Q,\alphabetA,q_0,F,\delta\rangle
\end{align*}
where $Q$ is a set of ""states@statesclassic"", $\alphabetA$ is an
alphabet (not necessarily finite), $q_0\in Q$ is the ""initial
state@initialclassic"", $F\subseteq Q$ is the set of ""final
states@finalclassic"", and $\delta_a\colon Q\rightarrow Q$ is the
""transition map@transition classic"" for all
letters~$a\in\alphabetA$.  The semantic of an automaton is to define
what is a ""run@runclassic"" over an input word~$u\in\alphabetA^*$,
and whether it is "accepting@acceptingclassic" or not.  Given a
word~$e=a_1\dots a_n$, the "automaton@dfaclassic"
""accepts@acceptsclassic"" the word
if~$\delta_{a_n}\circ\dots\circ\delta_{a_1}(q_0)\in F$, and otherwise
\reintro[acceptsclassic]{rejects it}.

If we see $q_0$ as a map~$\mathit{init}$ from the one element
set~$1=\{0\}$ to~$Q$, that maps~$0$ to~$q_0$, and $F$ as a map
$\mathit{final}$ from~$Q$ to the set~$2=\{0,1\}$, where~$1$ means
`accept' and $0$ means `reject', then the semantic of the
"automaton@dfaclassic" is to associate to each word~$u=a_1\dots a_n$
the map from~$1$ to~$2$ defined
as~$\mathit{final}\circ\delta_{a_n}\circ\cdots\circ\delta_{a_1}\circ\mathit{init}$. If
this map is (constant equal to) $1$, this means that the word is
accepted, and otherwise it is rejected.

\AP Pushing this idea further, we can see the semantics of the
automaton as a functor from the category $\intro*\defcatIword$ spanned
by the graph of vertices $\{\intro*\defobjectLeft, \intro*\defobjectCenter,\intro*\defobjectRight\}$ in Figure~\ref{fig:detAutAsFunct} to $\Set$, and more
precisely one that sends the object $\defobjectLeft$ to~$1$
and~$\defobjectRight$ to~$2$.
\AP
The arrows of the three-object category $\defcatIword$ are spanned by
$\intro*\defsymbLeft$, $\intro*\defsymbRight$ and $a$ for all
$a\in\alphabetA$, and the composite of $
\begin{tikzcd}
  \defobjectCenter\ar[r,"w"] & \defobjectCenter\ar[r,"w'"] & \defobjectCenter
\end{tikzcd}
$ is given by the concatenation $ww'$. Notice the left to right order
of composition.

In the above category, the arrows from $\defobjectLeft$ to $\defobjectRight$
are of the form $\defsymbLeft w\defsymbRight$ for $w$ an arbitrary word in
$\alphabetA^*$.
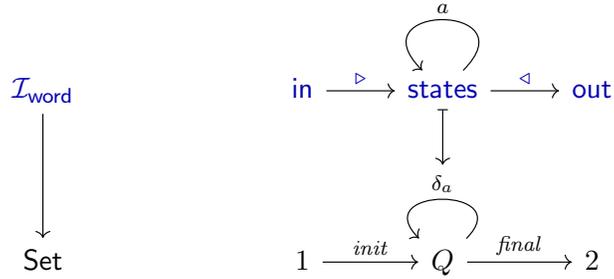
\begin{figure}
  \begin{tikzcd}
    \defcatIword\ar[dd] & \quad & \defobjectLeft\arrow[r,"\defsymbLeft"] &
    \defobjectCenter\arrow[loop,looseness=6,swap,
    "a"]\arrow[r,"\defsymbRight"]\ar[d,mapsto] & \defobjectRight \\
    & & & {} &\\
    \Set & \quad & 1\arrow[r,"\mathit{init}"] &
    Q\arrow[loop,looseness=4, swap,
    "\delta_a"]\arrow[r,"\mathit{final}"] & 2
  \end{tikzcd}
  \caption{Deterministic automata as functors}
  \label{fig:detAutAsFunct}
\end{figure}
\AP Furthermore, since a ""language@languageclassic"" can be seen as a
map from~$\alphabetA^*$ to the set~$1\rightarrow 2$ of functions from
$1$ to $2$, we can model it as a functor from the full subcategory
$\intro*\defcatOword$ on objects $\defobjectLeft$ and $\defobjectRight$ to the
category $\Set$, which maps $\defobjectLeft$ to~$1$ and~$\defobjectRight$
to~$2$.

\AP In this section we fix an arbitrary small category $\catI$ and a
full subcategory $\catO$. We denote by $\iota$ the inclusion functor
\[
  \begin{tikzcd}
    \catO\arrow[r,hook,"\iota"] & \catI\,.
  \end{tikzcd}
\]
We think of $\catI$ as a specification of the inner computations that
an automaton can perform, including black box behavior, not
observable from the outside.
On the other hand, the full subcategory $\catO$ specifies the
observable behavior of the automaton, that is, the language it
accepts.
In this interpretation, a machine/automaton $\autA$ is a functor from
$\catI$ to a category of outputs $\catC$, and the ``behavior'' or
``language'' of $\autA$ is the functor $\langL(\autA)$ obtained by
precomposition with the inclusion
$\begin{tikzcd} \catO\ar[hook]{r}{\iota}&\catI
\end{tikzcd}$. We obtain the following definition:
\noindent

\begin{defi}[$\catC$-languages and  $\catC$-automata]\label{def:aut-accepts-lang}\AP
  A ""$\catC$-language"" is a functor $\langL\colon\catO\to\catC$ and
  a ""$\catC$-automaton"" is a functor $\autA\colon\catI\to\catC$.  A
  "$\catC$-automaton" $\autA$ ""accepts"" a "$\catC$-language"
  $\langL$ when $\autA\circ\iota=\langL$; i.e. the following diagram
  commutes:
  \[
    \begin{tikzcd}
      \catO\ar{r}{\langL}\ar[hook]{d}[swap]{\iota} & \catC \\
      \catI\ar{ru}[swap]{\autA} &
    \end{tikzcd}
  \]
  We write $\intro*\catAutoL$ for the subcategory of the functor
  category $[\catI,\catC]$ where
  \begin{enumerate}
  \item objects are "$\catC$-automata" that "accept" $\langL$, and
  \item arrows are natural transformations $\alpha\colon\autA\to\autB$
    so that the natural transformation obtained by composition with
    the inclusion functor $\iota$ is the identity natural
    transformation on $\langL$, that is,
    $\alpha\circ \iota=\mathit{id}_\langL$.
  \end{enumerate}
\end{defi}

\begin{exa}[word automata and their
  languages]\label{ex:automata-as-functors}
  We can model various forms of word automata and their languages
  using the input categories
  $\begin{tikzcd} \defcatOword\ar[hook]{r}&\defcatIword
  \end{tikzcd}$ and varying the category of outputs:
  \begin{enumerate}
  \item As described in Figure~\ref{fig:detAutAsFunct}, deterministic
    automata can be seen as $\Set$-automata, i.e. as functors
    $\autA\colon\defcatIword\to\Set$ that map $\defobjectLeft$ to $1$ and
    $\defobjectRight$ to $2$.

    The language accepted by $\autA$ is the composite
    $\begin{tikzcd}
      \langL\colon\defcatOword\ar[hook]{r}&\defcatIword\ar[r,"\autA"] & \Set
    \end{tikzcd}$, which essentially specifies for each word
    $w\in A^*$ a function
    $\langL(\defsymbLeft w\defsymbRight)\colon 1\to 2$, establishing whether
    the word $w$ is accepted or not. Indeed, if $w=a_1\ldots a_n$,
    then $\langL(\defsymbLeft w\defsymbRight)$ is exactly the function
    $\mathit{final}\circ\delta_{a_n}\circ\cdots\circ\delta_{a_1}\circ\mathit{init}$
    described in the introduction of this section.
  \item Non-deterministic automata can be modeled as $\Rel$-automata,
    where $\Rel$ is the category whose objects are sets and maps are
    relations between them.

    \AP
    Indeed, a \intro{non-deterministic automaton} is completely
    determined by the relations described in the next diagram, where
    the set of initial states is modeled as a relation from $1$ to the
    set of states $Q$, the set of final states as a relation from $Q$
    to $1$ and the transition relation by any input letter $a$, as a
    relation on $Q$:
    \[
      \begin{tikzcd}
        1\arrow[r,negated,"\mathit{init}"] &
        Q\arrow[loop,looseness=6,negated,swap,
        "\delta_a"]\arrow[r,negated,"\mathit{final}"] & 1
      \end{tikzcd}
    \]
    Explicitly, we consider $\Rel$-automata
    $\autA\colon \defcatIword\to\Rel$ so that $\autA(\defobjectLeft)=1$ and
    $\autA(\defobjectRight)=1$.

    The language accepted by $\autA$ is the composite
    $\begin{tikzcd}
      \langL\colon\defcatOword\ar[hook]{r}&\defcatIword\ar[r,"\autA"] & \Rel
    \end{tikzcd}$. This functor specifies for each word $w\in A^*$ a
    relation
    $\begin{tikzcd}\langL(\defsymbLeft w\defsymbRight)\colon
      1\ar[r,negated]& 1
    \end{tikzcd}
    $. Notice that the set $\Rel(1,1)$ of relations on the set $1$ is
    isomorphic to $2$, and the relation
    $\langL(\defsymbLeft w\defsymbRight)$ simply models whether the word $w$
    is accepted by the automaton or not.

  \item Weighted automata over a semiring $S$ can be modeled as
    functors $\autA\colon\defcatIword\to\SMod$ valued in the category
    $\SMod$ of $S$-modules and $S$-linear morphisms and mapping both
    $\defobjectLeft$ and $\defobjectRight$ to $S$ (seen as a module over
    itself). Indeed, such an automaton is determined by the linear
    maps described in the next diagram, where the state space $Q$ has
    an $S$-module structure.
    \[
      \begin{tikzcd}
        S\arrow[r,"\mathit{init}"] & Q\arrow[loop,looseness=6,swap,
        "\delta_a"]\arrow[r,"\mathit{final}"] & S
      \end{tikzcd}
    \]
    Indeed, to give a linear map $\mathit{init}\colon S\to Q$ amounts
    to giving one element of the module $Q$, i.e. an initial state for
    the automaton.

    The language accepted by $\autA$, i.e. the composite
    $\begin{tikzcd}
      \langL\colon\defcatOword\ar[hook]{r}&\defcatIword\ar[r,"\autA"] &
      \SMod
    \end{tikzcd}$ specifies for each word $w\in A^*$ a linear
    transformation $\langL(\defsymbLeft w\defsymbRight)\colon S\to S$. Up to
    isomorphism, this is the same as specifying one scalar in $S$ for
    each word in $A^*$, hence we obtain the weighted language
    $A^*\to S$ classically accepted by the automaton.
  \end{enumerate}
\end{exa}

\subsection{Minimization of \texorpdfstring{"$\catC$-automata"}{C-automata}}
\label{sec:minim-catAuto}

In this section we show that the notion of a "minimal automaton" is an
instance of a more generic notion of minimal object that can be
defined in an arbitrary category $\catK$ whenever there exist an
""initial object"", a ""final object"", and a "factorization system"
$(\Epi,\Mono)$.

\AP Let $X,Y$ be two objects of $\catK$.  We say that:
\begin{center}
  $X$\quad$(\Epi,\Mono)$""-divides""\quad$Y$\qquad if\qquad $X$ is an
  \Epi""-quotient"" of an \Mono""-subobject"" of~$Y$.
\end{center}

Let us note immediately that in general this notion of
(\Epi,\Mono)"-divisibility" may not be transitive%
\footnote{There are nevertheless many situations for which it is the
  case; in particular when the category is "regular", and $\Epi$
  happens to be the class of "regular epis". This covers in particular
  the case of all "algebraic categories" with \Epi"-quotients" being
  the standard "quotients of algebras", and \Mono"-subobjects" being
  the standard "subalgebras".}.  \AP It is now natural to define an
"object"~$M$ to be (\Epi,\Mono)""-minimal"" in the "category", if it
(\Epi,\Mono)"-divides" all "objects" of the "category". Note that
there is no reason a priori that an (\Epi,\Mono)"-minimal" object in a
"category", if it exists, be unique up to "isomorphism". Nevertheless,
in our case, when the category has both "initial@initial object" and a
"final object", we can state the following minimization lemma:

\begin{lem}\AP\label{lemma:minimal}%
  Let~$\catK$ be a "category" with "initial object"~$I$ and "final
  object"~$F$ and let $(\Epi,\Mono)$ be a "factorization system"
  for~$\catK$. Define for every "object"~$X$:
  \begin{itemize}
  \item $\intro*\Min{}$ to be the "factorization" of the "unique
    arrow@initial object" from~$I$ to~$F$,
  \item $\intro*\Reach(X)$ to be the "factorization" of the "unique
    arrow@initial object" from~$I$ to~$X$, and $\intro*\Obs(X)$ to be
    the "factorization" of the "unique arrow@final object" from~$X$
    to~$F$.
  \end{itemize}
  Then
  \begin{itemize}
  \item $\Min{}$ is (\Epi,\Mono)"-minimal", and
  \item $\Min{}$ is "isomorphic" to both $\Obs(\Reach(X))$ and
    $\Reach(\Obs(X))$ for all "objects"~$X$.
  \end{itemize}
\end{lem}
\begin{proof} The proof essentially consists of a diagram:
  
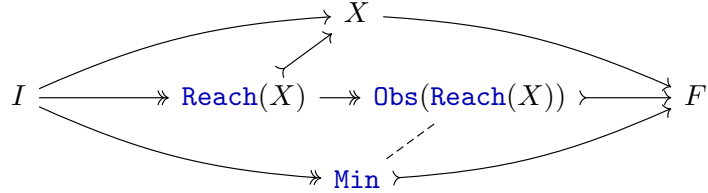
\begin{center}
    \begin{tikzcd}[column sep={1.5cm,between origins},row
      sep={1.1cm,between origins}]
      &&& X\ar[rrrd,bend left=10] & &
      \\
      I\ar[rrru,bend left=10]\ar[rrrd,bend right=10,two
      heads]\ar[rr,two heads] &&
      \Reach(X)\ar[ru,rightarrowtail]\ar[rr,two heads] &&
      \Obs(\Reach(X))\ar[rr, rightarrowtail] && F
      \\
      &&& \Min{} \ar[rrru,bend right=10,rightarrowtail]\ar[ru,dashed,
      no head]& &
    \end{tikzcd}
  \end{center}
  Using the definition of~$\Reach{}$ and~$\Obs{}$, and the fact that
  $\Epi$ is closed under "composition", we obtain that
  $\Obs(\Reach(X))$ is an (\Epi,\Mono)"-factorization" of the "unique
  arrow@initial object" from~$I$ to~$F$. Thus, thanks to the "diagonal
  property" of a "factorization system", $\Min{}$
  and~$\Obs(\Reach(X))$ are "isomorphic". Hence, furthermore,
  since~$\Obs(\Reach(X))$ (\Epi,\Mono)"-divides"~$X$ by construction,
  the same holds for~$\Min{}$.  In a symmetric way, we have the next diagram:
  
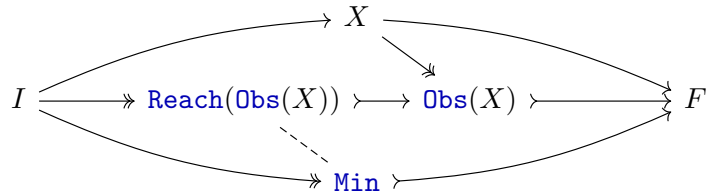
\begin{center}
    \begin{tikzcd}[column sep={1.5cm,between origins},row
      sep={1.1cm,between origins}]
      &&& X\ar[rrrd,bend left=10]\ar[rd,two heads] & &
      \\
      I\ar[rrru,bend left=10]\ar[rrrd,bend right=10,two
      heads]\ar[rr,two heads] &&
      \Reach(\Obs(X))\ar[rr,rightarrowtail] &&
      \Obs(X)\ar[rr, rightarrowtail] && F
      \\
      &&& \Min{} \ar[rrru,bend right=10,rightarrowtail]\ar[lu,dashed,
      no head]& &
    \end{tikzcd}
  \end{center}
  This shows that $\Reach(\Obs(X))$ is also "isomorphic" to~$\Min{}$.
\end{proof}

\AP An object $X$ of $\catK$ is called ""reachable"" when $X$ is
isomorphic to $\Reach(X)$. We denote by $\intro*\catReachK$ the full
subcategory of $\catK$ consisting of "reachable" objects. \AP
Similarly, an object $X$ of $\catK$ is called ""observable"" when $X$
is isomorphic to $\Obs(X)$. We denote by $\intro*\catObsK$ the full
subcategory of $\catK$ consisting of "observable" objects.

\AP We can express reachability $\Reach$ and observability $\Obs$ as
the right, respectively the left adjoint to the inclusion of
$\catReachK$, respectively of $\catObsK$ into $\catK$. It is indeed a
standard fact that factorization systems give rise to reflective
subcategories, see~\cite{cassidy_hebert_kelly_1985}. In our case, this
is the reflective subcategory $\catObsK$ of $\catK$. By a dual
argument, the category $\catReachK$ is coreflective in $\catK$. We can
summarize these facts in the next lemma.

\begin{lem}
  \label{lem:adj-reach-obs}
  Let~$\catK$ be a "category" with "initial object"~$I$ and "final
  object"~$F$ and let $(\Epi,\Mono)$ be a "factorization system"
  for~$\catK$.  We have the adjunctions
\end{lem}
\begin{equation}
  \label{eq:adj-reach-obs}
  \begin{tikzcd}[column sep={14mm,between origins}]
    \catReachK \ar[rr,bend left,hook] &\bot &\catK\ar[rr,bend left,
    "\Obs"]\ar[ll,bend left, "\Reach"] &\bot~ & \catObsK\,.\ar[ll,bend
    left, hook]
  \end{tikzcd}
\end{equation}

\AP
In what follows we will instantiate $\catK$ with the category
$\catAutoL$ of $\catC$-automata accepting a language
$\langL$. Assuming the existence of an ""initial@initial automaton""
and a ""final automaton"" for $\langL$ -- denoted by $\autAinitL$,
respectively $\autAfinalL$ -- and, of a factorization system, we
obtain the functorial version of the usual notions of "reachable
sub-automaton" $\ReachA$ and "observable quotient automaton" $\ObsA$
of an automaton $\autA$. The "minimal automaton" $\MinL$ for the
language $\langL$ is obtained via the factorization
\begin{equation}
  \label{eq:min-via-fact}
  \begin{tikzcd}
    \autAinitL \arrow[r,two heads] & \MinL\ar[r,tail]& \autAfinalL\,.
  \end{tikzcd}
\end{equation}

Lemma~\ref{lemma:minimal} implies that the "minimal automaton" divides
any other automaton recognizing the language, while a particular
instance of Lemma~\ref{lem:adj-reach-obs} pertaining to deterministic
automata is given in~\cite[Section~9.4]{BonchiBHPRS14}.

\begin{rem}
  The duality between reachability and observability can be stated as
  the duality between $\catReachK$ and $\catObsKop$. Indeed, if we
  consider the factorization system $(\Mono,\Epi)$ on
  $\catK^{\mathit{op}}$, then it immediately follows that $\catObsKop$
  is isomorphic to $\catReachK^\mathit{op}$. Hence the two adjunctions
  from~\eqref{eq:adj-reach-obs} are dual to each other.
\end{rem}

\begin{rem}[Minimization via adjunctions]
  As a consequence of Lemma~\ref{lemma:minimal}, minimization can be
  seen as an endofunctor $\Min\colon\catK\to\catK$, isomorphic to the
  functors obtained by considering any circuit in
  diagram~\eqref{eq:adj-reach-obs}.

  We will come back to this observation of regarding minimization via
  adjunctions, in Section~\ref{sec:brzozowski}, where we will show how
  Brzozowski's algorithms fits in the same conceptual approach,
  using however a longer chain of adjunctions.
\end{rem}

\subsection{Minimization of \texorpdfstring{"$\catC$-automata"}{C-automata}: sufficient conditions
  on \texorpdfstring{$\catC$}{C}}
\label{sec:minimization:suff-cond}
In this section we provide sufficient conditions on $\catC$ so that
the category $\catAutoL$ of "$\catC$-automata" accepting a
"$\catC$-language" $\langL$ satisfies the three conditions of
Lemma~\ref{lemma:minimal}. The sufficient conditions on $\catC$ are as
follows
\begin{enumerate}
\item completeness
\item cocompleteness
\item existence of a factorization system
\end{enumerate}

In Corollary~\ref{cor:exist-Kan-suff-cond} below we show that when
this conditions are satisfied then the initial and final automata for
a language exist and the minimal automaton can be obtained via the
factorization described in diagram~\eqref{eq:min-via-fact}.

\begin{rem}
  Before proceeding to the technical details, a few remarks are in
  order.
  \begin{enumerate}
  \item First, this notion of minimization is parametric in the
    factorization system one chooses on $\catC$. 
  \item Second, we emphasize that these conditions are only
    sufficient. In Section~\ref{sec:choffrut} we consider the example
    of sequential transducers and we instantiate $\catC$ with a
    "Kleisli category". Although this category is not complete,
    the final automaton exists.
  \item Finally, depending on the category $\catI$, we may relax the
    conditions in Corollary~\ref{cor:exist-Kan-suff-cond},
    see~Lemma~\ref{lem:the-minimization-wheel}. The reader may skip
    the rest of this section and consider
    Example~\ref{ex:minimization-dfa-wa}.
  \end{enumerate}
\end{rem}

We consider now the sufficient conditions on $\catC$ and we start with
the factorization system. It is well known that given a "factorization
system" $(\Epi,\Mono)$ on $\catC$, we can extend it to a
"factorization system" $(\Epi_{[\catI,\catC]},\Mono_{[\catI,\catC]})$
on the functor category $[\catI,\catC]$ in a point-wise fashion. That
is, a natural transformation is in~$\Epi_{[\catI,\catC]}$ if all its
components are in~$\Epi$, and analogously, a natural transformation is
in~$\Mono_{[\catI,\catC]}$ if all its components are in~$\Mono$. In
turn, the factorization system on the functor category $[\catI,\catC]$
induces a factorization system on its subcategory $\catAutoL$ for an
arbitrary language $\langL$.
\begin{lem}
  \label{lem:lifting-fact-syst}
  \AP If~$\catC$ has a factorization system~$(\Epi,\Mono)$,
  then~$\catAutoL$ has a factorization
  system~$(\intro*\EpiAutL,\intro*\MonoAutL)$, where~$\EpiAutL$
  consists of all the natural transformations with components
  in~$\Epi$ and~$\MonoAutL$ consists of all natural transformations
  with components in~$\Mono$.
\end{lem}
The proof of Lemma~\ref{lem:lifting-fact-syst} is the same as the
classical one that shows that "factorization systems" can be lifted to
functor categories.

As for the existence of the initial and final automaton accepting a
given language, we first notice that these can be stated in terms of
Kan extensions, see~\cite{maclane}.

\begin{lem}
  If the left Kan extension $\Lan{\langL}{\iota}$ of~$\langL$
  along~$\iota$ exists, then it is an initial object in~$\catAutoL$,
  that is, $\autAinitL$ exists and is isomorphic
  to~$\Lan{\langL}{\iota}$.

  Dually, if the right Kan extension~$\Ran{\langL}{\iota}$ of~$\langL$
  along~$\iota$ exists, then so does the final object~$\autAfinalL$
  of~$\catAutoL$ and $\autAfinalL$ is isomorphic
  to~$\Ran{\langL}{\iota}$.
\end{lem}

\begin{proof}[Proof Sketch]
  Assume the left Kan extension exists. Then the canonical natural
  transformation $\langL\to \Lan{\langL}{\iota}\circ \iota$ is an
  isomorphism since $\iota$ is full and faithful. Whenever $\autA$
  accepts $\langL$, that is, $\autA\circ\iota=\langL$, we obtain the
  required unique morphism $\Lan{\langL}{\iota}\to\autA$ using the
  universal property of the Kan extension. The argument for the right
  Kan extension follows by duality.
\end{proof}

\begin{cor}
  \label{cor:exist-Kan-suff-cond}
  Assume $\catC$ is complete, cocomplete and has a factorization
  system and let $\langL$ be a "$\catC$-language". Then the "initial
  $\langL$-automaton" and the "final $\langL$-automaton" exist and are
  given by the left, respectively right Kan extensions of $\langL$
  along $\iota$. Furthermore, the minimal $\catC$-automaton $\Min(L)$
  accepting $\langL$ is obtained via the factorization
  \[
    \begin{tikzcd}
      \Lan{\langL}{\iota}\ar[r,two heads] & \Min(L)\ar[r,tail] &
      \Ran{\langL}{\iota}\,.
    \end{tikzcd}
  \]
\end{cor}


\section{Word Automata}
\label{sec:word-input-category}

In Sections~\ref{sec:word-input-category} to~\ref{sec:brzozowski} we
restrict our attention to the case of word automata, for which we
recall the input category $\intro*\catIword$ from
Figure~\ref{fig:detAutAsFunct} (i.e., the three-object category with
arrows spanned by $\intro*\symbLeft$, $\intro*\symbRight$ and $a$ for
all $a\in\alphabetA$) and its full subcategory $\intro*\catOword$ on objects
$\objectLeft$ and $\objectRight$.

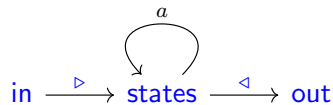
\AP\noindent
\phantomintro\objectLeft\phantomintro\objectRight\phantomintro\objectCenter
\[
  \begin{tikzcd}
    \reintro*\objectLeft\arrow[r,"\reintro*\symbLeft"] &
    \reintro*\objectCenter\arrow[loop,looseness=6,swap,
    "a"]\arrow[r,"\reintro*\symbRight"] & \reintro*\objectRight
  \end{tikzcd}
\]

\AP 
We consider "$\catC$-languages", which are now functors
$\langL\colon\catOword\to\catC$.  If $\langL(\objectLeft)=X$ and
$\langL(\objectRight)=Y$ we call $\langL$ a
""$(\catC, X,Y)$-language"".  Similarly, we consider
"$\catC$-automata" that are functors $\autA\colon\catIword\to\catC$.
If $\autA(\objectLeft)=X$ and $\autA(\objectRight)=Y$ we call $\autA$
a ""$(\catC, X,Y)$-automaton"".

\subsection{Minimization of word automata} We first provide a couple
of instances of the generic minimization results given in the previous
section and then we show how they can be refined considering the
particular structure of the category $\catIword$.

\begin{exa}
  \label{ex:minimization-dfa-wa}
  \begin{enumerate}
  \item \textit{Deterministic automata, i.e., $(\Set,1,2)$-automata.}
    Since $\Set$ is complete and cocomplete, the initial and final
    automaton accepting a language $\langL$ can be computed as in
    Corollary~\ref{cor:exist-Kan-suff-cond}.
    \begin{itemize}
    \item The initial automaton $\autAinitL$ is described in the next
      diagram.
      \[
        \begin{tikzcd}
          1\arrow[r,"\reintro*\varepsilon"] &
          A^*\arrow[loop,looseness=6,swap,
          "\delta_a"]\arrow[r,"\langL?"] & 2
        \end{tikzcd}
      \]
      Its state space is the set $A^*$ of all words, the initial state
      is the empty word $\varepsilon$, the set of final states consist of
      the words belonging to the accepted language, and, for each
      input letter $a$, the transition map $\delta_a$ is defined by
      $w\mapsto wa$.

    \item The final automaton $\autAfinalL$ is described in the next
      diagram.
      \[
        \begin{tikzcd}
          1\arrow[r,"\reintro*\langL"] &
          2^{A^*}\arrow[loop,looseness=6,swap,
          "\delta_a"]\arrow[r,"\varepsilon ?"] & 2
        \end{tikzcd}
      \]

      Its state space is the set $2^{A^*}$ of all languages, the
      initial state is the language accepted by the automaton, the
      final state consists of all the languages that contain the empty
      word, and, for each input letter $a\in A$, the transition map
      $\delta_a$ is taking left quotients of a language by the letter
      $a$, that is,
      $A^*\ni K\mapsto a^{-1}K=\{u\in A^*\mid au\in K\}$.
    \item The factorization system of $\catAutoL$ is inherited from
      $\Set$, consisting of surjective, respectively injective
      functions.
    \end{itemize}

    Indeed, $\autAinitL$ and $\autAfinalL$ are the initial,
    respectively the final objects in the category $\catAutoL$. To see
    this, notice that for any other automaton $\autA$ in $\catAutoL$
    with state space $\autA(\objectCenter)=Q$, we have a unique
    morphism from $\autAinitL$ to $\autA$, as described in the next
    diagram on the left. This is a natural transformation, determined
    by its component on the object $\objectCenter$, i.e., by the
    function $\mathsf{reachedState}\colon A^*\to Q$, which maps a word
    $w\in A^*$ to the state of $Q$ reached by reading $w$.  Similarly,
    there is a unique automata morphism from $\autA$ to $\autAfinalL$,
    determined by the function
    $\mathsf{acceptedLanguage}\colon Q\to 2^{A^*}$ which maps a state
    $q\in Q$ to the language accepted by $q$.

    In particular, as shown in the next diagram (on the right), the
    unique morphism from $\autAinitL$ to $\autAfinalL$ is determined
    by the function $A^*\to 2^{A^*}$ which maps a word $w\in A^*$ to
    the quotient language $w^{-1}\langL$. If we factorize this map, we
    obtain the quotient of $A^*$ by the syntactic equivalence $\sim$
    defined by $w\sim w'$ if and only if $w^{-1}\langL=w'^{-1}\langL$,
    that is, we obtain the state space of the minimal automaton
    accepting $\langL$.

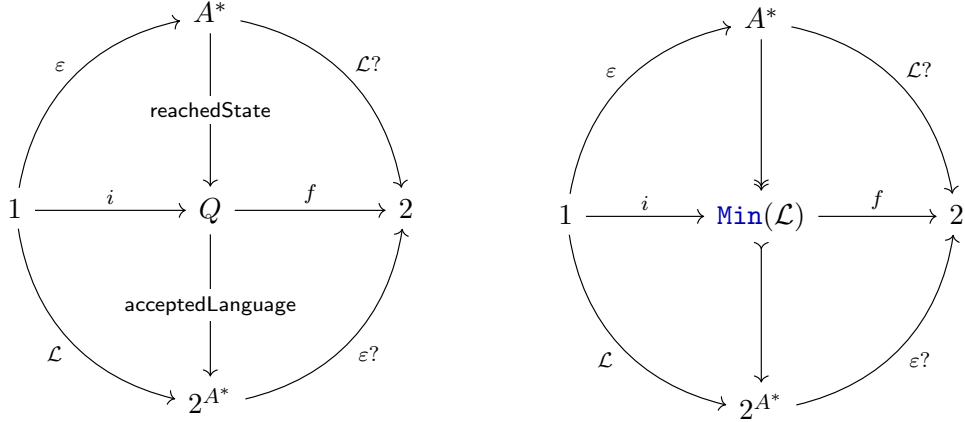
    \begin{minipage}{.5\textwidth}
      \[
        \begin{tikzcd}[column sep={2.6cm,between origins},row
          sep={1.3cm,between origins}]
          & A^*
          \arrow[rdd,bend left=34,
          "\langL?"]\arrow[dd,"\mathsf{reachedState}" description] &
          \\
          \\
          1
          \arrow[bend right=34,swap,"\langL"]{rdd} \arrow[ruu,,bend
          left=34, "\varepsilon"] \arrow[r,"i"] & Q \arrow[r,"f"]
          \arrow[dd,"\mathsf{acceptedLanguage}" description]
          & 2
          \\
          \\
          &
          2^{A^*}\arrow[bend right=34]{ruu}[swap]{\varepsilon ?} &  \\
        \end{tikzcd}
      \]
    \end{minipage}
    \begin{minipage}{.45\textwidth}
      \[
        \begin{tikzcd}[column sep={2.6cm,between origins},row
          sep={1.3cm,between origins}]
          & A^*
          \arrow[rdd,bend left=34, "\langL?"]\arrow[dd,two heads] &
          \\
          \\
          1
          \arrow[bend right=34,swap,"\langL"]{rdd} \arrow[ruu,,bend
          left=34, "\varepsilon"] \arrow[r,"i"] & \MinL \arrow[r,"f"]
          \arrow[dd,tail]
          & 2
          \\
          \\
          &
          2^{A^*}\arrow[bend right=34]{ruu}[swap]{\varepsilon ?} &  \\
        \end{tikzcd}
      \]
    \end{minipage}
  \item \textit{Non-deterministic automata, i.e.,
      $(\Rel,1,1)$-automata.}  The category $\Rel$ has countable
    products and coproducts, so the initial and final automata do
    exist and both have as state space the set $A^*$. However, the
    missing ingredient for obtaining a meaningful notion of minimal
    automaton in this case is a suitable factorization system.
  \item \textit{Weighted automata over a field $K$, i.e.,
      $(\KMod,K,K)$-automata.} We have a similar situation depicted in
    the two diagrams below. The initial automaton has as state space
    the vector space of finitely supported functions from $A^*\to K$,
    while the final automaton has as state space the vector space of
    all functions $A^*\to K$. These are precisely the coproduct,
    respectively the product of $A^*$ many copies of $K$ in the
    category $\KMod$. In the diagram below, the linear transformations
    $\varepsilon$ and $\varepsilon?$ are precisely the injection into
    the coproduct, respectively the projection from the product which
    correspond to the $\varepsilon$-component. The linear
    transformation, denoted (by an abuse) by $\langL$, maps the unit
    of the field to the weighted language $\langL\colon A^*\to K$
    accepted by the automaton. While the linear map $\langL?$ can be
    defined on a basis of $\bigoplus\limits_{u\in A^*}K$ as follows:
    for a given $u\in A^*$, it maps the unit of $K$ of the
    $u$-component of the coproduct to $\langL(u)$. If we factorize the
    unique linear transformation from $ \bigoplus\limits_{u\in A^*}K$
    to $\prod\limits_{u\in A^*}K$ we obtain precisely the vector space of
    the minimal automaton accepting $\langL$.
  
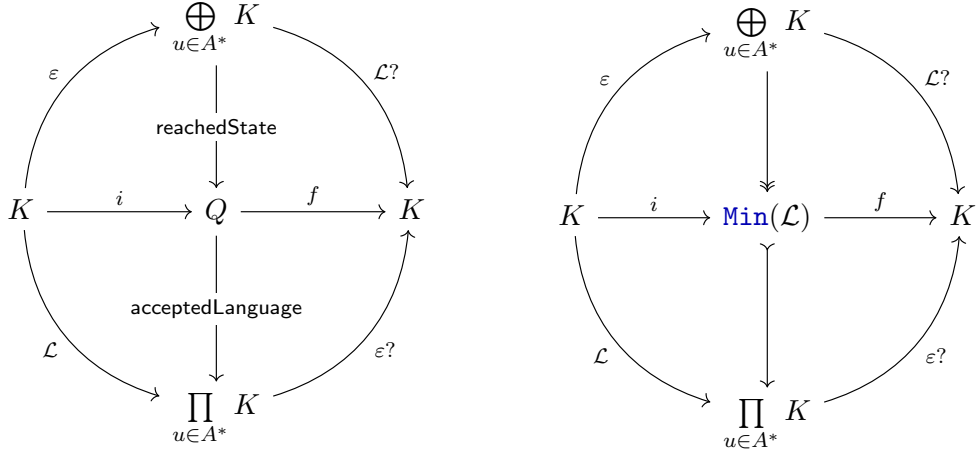
    \begin{minipage}{.5\textwidth}
      \[
        \begin{tikzcd}[column sep={2.6cm,between origins},row
          sep={1.3cm,between origins}]
          & \bigoplus\limits_{u\in A^*}K
          \arrow[rdd,bend left=34,
          "\langL?"]\arrow[dd,"\mathsf{reachedState}" description] &
          \\
          \\
          K
          \arrow[bend right=34,swap,"\langL"]{rdd} \arrow[ruu,,bend
          left=34, "\varepsilon"] \arrow[r,"i"] & Q \arrow[r,"f"]
          \arrow[dd,"\mathsf{acceptedLanguage}" description]
          & K
          \\
          \\
          &
          \prod\limits_{u\in A^*}K\arrow[bend right=34]{ruu}[swap]{\varepsilon ?} &  \\
        \end{tikzcd}
      \]
    \end{minipage}
    \begin{minipage}{.45\textwidth}
      \[
        \begin{tikzcd}[column sep={2.6cm,between origins},row
          sep={1.3cm,between origins}]
          & \bigoplus\limits_{u\in A^*}K
          \arrow[rdd,bend left=34, "\langL?"]\arrow[dd,two heads] &
          \\
          \\
          K
          \arrow[bend right=34,swap,"\langL"]{rdd} \arrow[ruu,,bend
          left=34, "\varepsilon"] \arrow[r,"i"] & \MinL \arrow[r,"f"]
          \arrow[dd,tail]
          & K
          \\
          \\
          &
          \prod\limits_{u\in A^*}K\arrow[bend right=34]{ruu}[swap]{\varepsilon ?} &  \\
        \end{tikzcd}
      \]
    \end{minipage}
  \end{enumerate}
\end{exa}

The examples above are instances of the next generic lemma, which
refines the statement of Corollary~\ref{cor:exist-Kan-suff-cond}
taking into account the particular structure of the input category
$\catIword$.

\begin{lem}[from~\cite{DBLP:conf/mfcs/ColcombetP17}]
  \label{lem:the-minimization-wheel}
  If $\catC$ has countable products and countable coproducts, and a
  factorization system, then the minimal "$\catC$-automaton"
  "accepting" $\langL$ is obtained via the "factorization" in the next
  diagram.
  
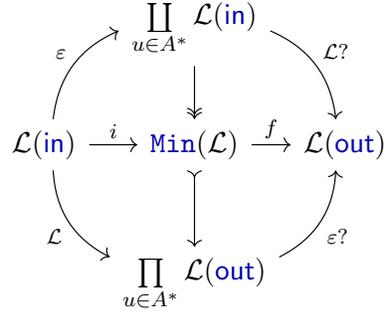
\[
    \begin{tikzcd}[column sep={2cm,between origins},row
      sep={1.7cm,between origins}]
      & \coprod\limits_{u\in A^*}\langL(\objectLeft)
      \arrow[rd,bend left, "\langL ?"]\arrow[d,two heads] &
      \\
      \langL(\objectLeft)
      \arrow[bend right,swap,"\langL"]{rd} \arrow[ru,,bend left,
      "\varepsilon"] \arrow[r,"i"] & \MinL \arrow[r,"f"] \arrow[d,
      tail] & \langL(\objectRight)
      \\
      & \prod\limits_{u\in A^*}\langL(\objectRight) \arrow[bend
      right]{ru}[swap]{\varepsilon ?} &
    \end{tikzcd}
  \]
\end{lem}

The "initial automaton" $\autAinitL$ has as state space the copower
$\coprod\limits_{u\in A^*}\langL(\objectLeft)$. The map
\[
  \varepsilon=\autAinitL(\symbLeft)\colon\langL(\objectLeft)\to\coprod\limits_{u\in
    A^*}\langL(\objectLeft)
\] is the coproduct injection corresponding to $\varepsilon\in
A^*$. The map
\[
  \langL?=\autAinitL(\symbRight)\colon\coprod\limits_{u\in
    A^*}\langL(\objectLeft)\to\langL(\objectRight)
\] is given on the component of the coproduct corresponding to
$u\in A^*$ by $\langL(\symbLeft u\symbRight)$. Lastly, for each
$a\in A$ the map $\autAinitL(a)$ is given on the component of the
coproduct that corresponds to $u\in A^*$ as the coproduct injection
corresponding to the word $ua$.

The "final automaton" $\autAfinalL$ can be obtained by a duality
argument. It has as state space the power
$\prod\limits_{u\in A^*}\langL(\objectRight)$. The map
\[
  \langL=\autAfinalL(\symbLeft)\colon\langL(\objectLeft)\to\prod\limits_{u\in
    A^*}\langL(\objectRight)
\] is obtained using the universal property of the product by
considering for each $u\in A^*$ the map
$\langL(\symbLeft
u\symbRight)\colon\langL(\objectLeft)\to\langL(\objectRight)$.
The map
\[
  \varepsilon?=\autAfinalL(\symbRight)\colon\prod\limits_{u\in
    A^*}\langL(\objectRight)\to\langL(\objectRight)
\] is the projection corresponding to the $\varepsilon$ component of the
product.  Lastly, for each $a\in A$ the map $\autAfinalL(a)$ is
obtained by taking the product by $u\in A^*$ of the projections
$\prod\limits_{u\in A^*}\langL(\objectRight)\to\langL(\objectRight)$ on
the $au$ component.

In~\cite{DBLP:conf/mfcs/ColcombetP17} we gave a direct proof for the
initiality of $\autAinitL$. Here we can also notice that this is
exactly the result of colimit computation of the left Kan extension of
$\langL$ along $\iota$ mentioned in
Corollary~\ref{cor:exist-Kan-suff-cond}. Indeed, we can use the fact
that there are no morphisms from $\objectRight$ to $\objectCenter$ in
$\catI$ and the only morphism on which you take the colimit are of the
form $\symbLeft w\colon \objectLeft\to\objectCenter$ for all
$w\in\alphabetA^*$.

For the final automaton, the proof follows by duality.

\begin{rem}
  \label{rem:comp-well-pointed-coalgebras}
  When the category $\catC$ has copowers, then a
  "$(\catC, X,Y)$-automaton" gives rise to a pair $(\alpha,f)$
  consisting of an algebra $\alpha\colon FQ\to Q$ for the functor
  $F\colon\catC\to\catC$, $FZ=X+ A\cdot Z$ and of a morphism
  $f\colon Q\to Y$.
  
  Dually, when the category $\catC$ has powers, then a
  "$(\catC, X,Y)$-automaton" gives rise to a pair $(\xi,i)$ consisting
  of a coalgebra $\xi\colon Q\to HQ$ for the functor
  $H\colon\catC\to\catC$, $HZ=Y\times Z^A$ and of a morphism
  $i\colon X\to Q$. This is a mild generalization of the notion of
  pointed colagebras, see e.g.~\cite{AdamekMMS:Well-pointed-coalg},
  where $\catC$ is assumed concrete over $\Set$ and the map $i$
  corresponds to selecting an element in
  $Q$. In~\cite{AdamekMMS:Well-pointed-coalg} minimal automata are
  seen as \emph{well-pointed coalgebras}. However, in
  Section~\ref{sec:choffrut} we will see an example of automata in a
  "Kleisli category" that does not have powers, and where the
  equivalence between the functorial approach and the coalgebraic one
  breaks.
\end{rem}

\subsection{Lifting Adjunctions to Categories of Automata}
\label{sec:lifting-adjucntions-to-automata}

In Example~\ref{ex:automata-as-functors} we have seen how languages
$L\subseteq A^*$ can be modeled as functors from $\catOword$ to
either $\Set$ or $\Rel$, using the fact that
$\Set(1,2)\simeq\Rel(1,1)\simeq 2$. We will see how the relationship
between deterministic and non-deterministic automata, i.e.,
"$(\Set,1,2)$-automata" and "$(\Rel,1,1)$-automata", can be derived from
the relationship between the categories $\Set$ and $\Rel$, and, is a
special instance of a more general phenomenon. This is the subject of
the present section, in which we will juggle with "languages" and
"automata" interpreted over different output categories connected via
"adjunctions".

Assume we have an adjunction between two categories $\catC$ and
$\catD$
\knowledgeconfigure{quotation=false}%
\[
  \begin{tikzcd}
    \catC \ar[rr,bend left=14, "F"] & \bot &\catD\ar[ll,bend left=14,
    "G"]\,,
  \end{tikzcd}
\]
\knowledgeconfigure{quotation=true}%
with $F\dashv G\colon \catD\to\catC$.  Let $(-)^*$ and $(-)_*$ denote
the induced natural isomorphisms between the homsets.  In particular,
given objects $I$ in $\catC$ and $O$ in $\catD$, we have bijections
\knowledgeconfigure{quotation=false}%
\begin{equation}
  \label{eq:bij-lang}
  \begin{tikzcd}[column sep=large]
    \catC(I,GO)\ar[rr,shift left=2, "(-)^*"] &
    &\catD(FI,O)\ar[ll,shift left=2, "(-)_*"]
  \end{tikzcd}
\end{equation}
\knowledgeconfigure{quotation=true}%
These bijections induce a one-to-one correspondence between
"$(\catC,I,GO)$-languages" and "$(\catD,FI,O)$-languages", which by an
abuse of notation we denote by the same symbols:
\[
  \begin{tikzcd}[column sep=large]
    \text{\kl{$(\catC,I,GO)$-languages}}\ar[rr, shift left=2, "(-)^*"]
    & &\text{\kl{$(\catD,FI,O)$-languages}}\ar[ll,shift left=2,
    "(-)_*"]
  \end{tikzcd}
\]

Indeed, given a "$(\catC,I,GO)$-language" $\langL\colon \catO\to\catC$
we obtain a "$(\catD,FI,O)$-language" $\langL^*\colon \catO\to\catD$
by setting
$ \langL^*(\symbLeft w\symbRight) = (\langL(\symbLeft w\symbRight))^*
\in \catD(FI,O)$. Conversely, given a "$(\catD,FI,O)$-language"
$\langL'$ we obtain a "$(\catC,I,GO)$-language" $(\langL')_*$ by
setting
$(\langL')_*(\symbLeft w\symbRight)=(\langL'(\symbLeft
w\symbRight))_*$.

The next lemma shows how we can lift an adjunction between the output
categories $\catC$ and $\catD$ to an adjunction between categories of
automata that accept essentially the same language, admitting two
equivalent representations in $\catC$ and $\catD$.

\begin{lem}
  \label{lem:lifting-adjunctions}
  Assume $\langL_\catC$ and $\langL_\catD$ are "$(\catC,I,GO)$-@language",
  respectively "$(\catD,FI,O)$-languages" so that
  $\langL_\catD=(\langL_\catC)^*$. Then the adjunction $F\dashv G$
  lifts to an adjunction
  $\overline{F}\dashv\overline{G}\colon\catAutoLD\to\catAutoLC$. The
  lifted functors $\overline{F}$ and $\overline{G}$ are defined as
  $F$, resp. $G$ on the state object, that is, the following diagram
  commutes
  \knowledgeconfigure{quotation=false}%
  \[
    \begin{tikzcd}[column sep=large,row sep=5mm]
      \catAutoLC\ar[dd,"\St"] \ar[rr,bend left=12, "\overline{F}"] & \bot &\catAutoLD\ar[ll,bend left=12, "\overline{G}"]\ar[dd,"\St"]\\
      & & \\
      \catC \ar[rr,bend left=14, "F"] & \bot &\catD\ar[ll,bend
      left=14, "G"]
    \end{tikzcd}
  \]
  \knowledgeconfigure{quotation=true}%
  where the functor $\St\colon\catAutoLC\to\catC$ is the evaluation at
  $\objectCenter$, that is, it sends an automaton
  $\autA\colon\catIword\to\catC$ to $\autA(\objectCenter)$.
\end{lem}
\begin{proof}[Proof sketch]
  The functor $\overline{F}$ maps an "$\catC$-automaton"
  $\autA\colon\catIword\to\catC$ from $\catAutoLC$ to the
  $\catD$-automaton $\overline{F}\autA\colon \catIword\to\catD$ mapping
  $\symbLeft\colon\objectLeft\to\objectCenter$ to
  $F(\autA(\symbLeft))$, $a\colon\objectCenter\to\objectCenter$ to
  $F(\autA(a))$ and $\symbRight\colon\objectCenter\to\objectRight$ to
  the adjoint transpose
  $(\autA(\symbRight))^*\colon F\autA(\objectCenter)\to O$ of
  $\autA(\symbRight)\colon\autA(\objectCenter)\to GO$. In a diagram
  \[
    \begin{tikzcd}
      I\arrow[r,"\autA(\symbLeft)"] &
      \autA(\objectCenter)\arrow[loop,looseness=6,swap,
      "\autA(a)"]\arrow[r,"\autA(\symbRight)"] & GO
      &{}\arrow[r,mapsto,"\overline{F}"] & {} & FI
      \arrow[r,"F(\autA(\symbLeft))"] &
      F\autA(\objectCenter)\arrow[loop,looseness=6,swap,
      "F(\autA(a))"]\arrow[r,"(\autA(\symbRight))^*"] & O
    \end{tikzcd}
  \]
  The functor $\overline{G}$ is defined similarly on an
  "$\catD$-automaton" $\autB$.
  \[
    \begin{tikzcd}
      FI\arrow[r,"\autB(\symbLeft)"] &
      \autB(\objectCenter)\arrow[loop,looseness=6,swap,
      "\autB(a)"]\arrow[r,"\autB(\symbRight)"] & O
      &{}\arrow[r,mapsto,"\overline{G}"] & {} & I
      \arrow[r,"(\autB(\symbLeft))_*"] &
      G\autB(\objectCenter)\arrow[loop,looseness=6,swap,
      "G(\autB(a))"]\arrow[r,"G(\autB(\symbRight))"] & GO
    \end{tikzcd}
  \]
  We show next that we have an isomorphism
  \[
    \catAutoLD(\overline{F}\autA,\autB)\cong\catAutoLC(\autA,\overline{G}\autB)
  \]
  Indeed, consider a morphism $\alpha\colon\overline{F}\autA\to\autB$
  in $\catAutoLD$.  We define a natural transformation
  $\alpha_*\colon \autA\to\overline{G}\autB$ by setting its component
  at $\objectCenter$ as the adjoint transpose
  $(\alpha_{\objectCenter})_*$ of
  \[\alpha_{\objectCenter}\colon
    F\autA(\objectCenter)\to\autB(\objectCenter)\,.\] It is now easy
  to verify that $\alpha_*$ is indeed an automata morphism in
  $\catAutoLC$ and that the mapping $\alpha\mapsto\alpha_*$ gives rise
  to the desired isomorphism.
\end{proof}

\subsection{Application: non-deterministic automata and
  (co)determinization}
\label{sec:det-codet-adj}

As a first application of Lemma~\ref{lem:lifting-adjunctions}, we see
how determinization of non-deterministic automata can be seen as a
right adjoint to the inclusion of deterministic automata into
non-deterministic ones. Similarly, codeterminization is a left adjoint
to the inclusion of codeterministic automata into non-deterministic
ones.



\AP
Given a language $L\subseteq A^*$ we can model it in several
equivalent ways: as a "$(\Set,1,2)$-language" $\intro*\langLSet$, or as a
"$(\Setop,2,1)$-language" $\intro*\langLSetop$, or, lastly as a
"$(\Rel,1,1)$-language" $\intro*\langLRel$. This is because we can model the
fact $w\in L$ using a morphisms in either of the three isomorphic
homsets
\begin{equation}
  \label{eq:iso-homsets}
  \Set(1,2)\cong \Setop(2,1)\cong \Rel(1,1)\cong 2\,.
\end{equation}

These isomorphisms, can be seen in turn as arising from the next two
adjunctions:
\[
  \begin{tikzcd}
    \Set \ar[rr,bend left=20, "\FPow"] & \bot &\Rel\ar[ll,bend
    left=20, "\UPow"] \ar[rr,bend left=20, "\UPowop"]& \bot & \Setop
    \ar[ll,bend left=20, "\FPowop "]
  \end{tikzcd}
\]
\AP The adjunction between $\Set$ and $\Rel$ is the Kleisli adjunction
for the powerset monad: $\intro*\FPow$ is identity on objects as maps
a function $f\colon X\to Y$ to itself $f\colon X\tokl Y$, but seen as
a relation. The functor $\intro*\UPow$ maps $X$ to its powerset
$\Pow(X)$, and a relation $R\colon X\to Y$ to the function
$\UPow(R)\colon\Pow(X)\to\Pow(Y)$ mapping $A\subseteq X$ to
$\{y\in Y\mid\exists x\in X. (x,y)\in R\}$.
\phantomintro\UPowop\phantomintro\FPowop The adjunction between
$\Setop$ and $\Rel$ is the dual of the previous one, composed with the
self-duality of $\Rel$.

The isomorphisms of homsets
from~\eqref{eq:iso-homsets} can be rephrased as
\[
  \Set(1,\UPow 1)\cong\Rel(\FPow 1,1)\textrm{\quad and \quad} \Setop(\UPowop 1,1)\cong\Rel(1,\FPowop 1)\,.
  \]
  In particular, we can regard $\langLSet$ above as a
  $(\Set,1,\UPow 1)$-language and $\langLRel$ as a
  $(\Rel,\FPow 1,1)$-language. Using the notations from
  Lemma~\ref{lem:lifting-adjunctions}, we have that
  $\langLRel=(\langLSet)^*$. Similarly, we obtain that
  $\langLRel=(\langLSetop)_*$.

  Determinization and codeterminization (without restriction to
  reachable states as in the operations $\determinize$ and
  $\codeterminize$ introduced in Section~\ref{sec:brzozowski}) of a
  $\Rel$-automaton can be seen as applications of
  Lemma~\ref{lem:lifting-adjunctions} and are obtained by lifting the
  adjunctions between $\Set$, $\Rel$ and $\Setop$ as in the next
  diagram. The left adjoint $\overline\FPow$ transforms a
  deterministic automaton into a non-deterministic one, while the
  right adjoint $\overline{\UPow}$ is the ""determinization
  functor"". On the other hand, the left adjoint functor
  $\overline{\UPowop}$ is the ""codeterminization functor"".

\begin{equation}
  \label{eq:2}
  \begin{tikzcd}
    \catAutoLSet\ar[ddd,""] \ar[rr,bend left=15, "\overline{\FPow}"] &
    \bot &\catAutoLRel\ar[ll,bend left=15,
    "\overline{\UPow}"]\ar[ddd,""]
    \ar[rr,bend left=15, "\overline{\UPowop}"] & \bot & \catAutoLSetop \ar[ll,bend left=15, "\overline{\FPowop}"]\ar[ddd," "]\\
    & & \\
    & & \\
    \Set \ar[rr,bend left=20, "\FPow"] & \bot &\Rel\ar[ll,bend
    left=20, "\UPow"] \ar[rr,bend left=20, "\UPowop"]& \bot & \Setop
    \ar[ll,bend left=20, "\FPowop "]
  \end{tikzcd}
\end{equation}

\section{Choffrut's minimization of subsequential transducers}
\label{sec:choffrut}

In~\cite{Choffrut79,Choffrut2003} Choffrut establishes a minimality
result for "subsequential transducers", which are deterministic
automata that output a word while processing their input.  In this
section, we show that this result can be established in the functorial
framework of this paper.

We first present the model of "subsequential transducers" in
Section~\ref{subsection:subsequential}, show how these can be
identified with "automata in the Kleisli category of a suitably chosen
monad", and state the minimization result,
Theorem~\ref{theorem:minimization transducer}. The subsequent sections
provide the necessary material for proving the theorem.

\subsection{Subsequential transducers and automata in a Kleisli
  category}
\label{subsection:subsequential}

"Subsequential transducers" are (finite state) machines that compute
partial functions from input words in some alphabet~$\alphabetA$ to
output words in some other alphabet~$\alphabetB$.  In this section, we
recall the classical definition of these objects, and show how it can
be phrased categorically.

\begin{defi}
  A ""subsequential transducer"" is a tuple
  \begin{align*}
    T=(Q, \alphabetA, \alphabetB, q_0, t, u_0,(-\cdot
    a)_{a\in\alphabetA},(-*a)_{a\in\alphabetA})\ ,
  \end{align*}
  where
  \begin{itemize}
  \item $\alphabetA$ is the ""input alphabet@input alphabetSS"" and
    $\alphabetB$ the ""output one@output alphabetSS"",
  \item $Q$ is a (finite) set of ""states@statesSS"".
  \item $q_0$ is either undefined or belongs to~$Q$ and is called the
    ""initial state of the transducer"".
  \item $t\colon Q\rightharpoonup \alphabetB^*$ is a ""partial
    termination function"".
  \item $u_0\in \alphabetB^*$ is either undefined and is defined if
    and only if~$q_0$ is, and is the ""initialization value"".
  \item $-\cdot a\colon Q\rightharpoonup Q$ is the ""partial
    transition function for the letter~$a$"", for all $a\in A$.
  \item $-* a\colon Q\rightharpoonup \alphabetB^*$ is the ""partial
    production function for the letter~$a$"" for all
    $a\in \alphabetA$; it is required that $q* a$ be defined if and
    only if $(q\cdot a)$ is.
  \end{itemize}
  The "subsequential transducer" computes a partial function
  $\intro*\semTrans T\colon \alphabetA^*\rightharpoonup \alphabetB^*$
  defined as:
  \begin{align*}
    \semTrans T(a_1\dots a_n)&=u_0(q_0*a_1)(q_1*a_2)\dots
                               (q_{n-1}*a_n)t(q_n)&\text{for all~$a_1\dots a_n\in \alphabetA^*$,}
  \end{align*}
  where each~$q_i$ is either undefined or belongs to~$Q$, with $q_0$
  inherited from the definition of~$T$, and $q_i=q_{i-1}\cdot a_i$ for
  all~$i=1\dots n$.
\end{defi}

\AP These "subsequential transducers" are modeled in our framework as
"automata in the category" of free algebras for the
monad~$\monadTrans$, that we describe now.
\begin{defi}
  The monad $\intro*\monadTrans\colon\Set\to\Set$ is defined by
  \[
    \monadTrans X=B^*\times X + 1
  \]
  with unit $\eta_X$ and multiplication $\mu_X$ defined for all
  $x\in X$ and $w,u\in B^*$ as:
  \begin{align*}
    &&
       \mu_X\colon\qquad\monadTrans^2 X&\to \monadTrans X\\
    \eta_X\colon\quad X &\to B^*\times X +1&
                                             (w,(u,x))&\mapsto (wu, x)\\
    x&\mapsto (\varepsilon, x)&
                                (w,\bot) &\mapsto \bot\\
    && \bot &\mapsto \bot
  \end{align*}
  where we denote by $\bot$ the unique element of $1$ (used to model
  the partiality of functions).
\end{defi}

\AP Recall that the category of free $\monadTrans$-algebras, i.e., the
""Kleisli category"" $\intro*\KlTrans$ for $\monadTrans$, has as
objects sets $X,Y, \ldots$ and as ""morphisms@Kleisli morphism""
$f\colon X\tokl Y$ functions $f\colon X\to B^*\times X+1$ in $\Set$
(that is partial functions from $X$ to $B^*\times Y$).

\AP Let $T$ be a "subsequential transducer". The "initial state of the
transducer" $q_0$ and the "initialization value" $u_0$ together form a
morphism $i\colon 1\tokl Q$ in the category $\KlTrans$. Similarly, the
"partial transition function" and the "partial production function"
for a letter $a$ of the input alphabet $A$ are naturally identified to
"Kleisli morphisms" $\delta_a\colon Q\tokl Q$ in $\KlTrans$. Finally,
the "partial termination function" together with the "partial
production function" are nothing but a "Kleisli morphism" of the form
$t\colon Q\tokl 1$. To summarize, we obtained that a "subsequential
transducer" $T$ in the sense of~\cite{Choffrut2003} is specified by
the following "morphisms@Kleisli morphism" in $\KlTrans$:
\[
  \begin{tikzcd}
    1\arrow[r,negated,"i"] & Q\arrow[loop,looseness=6,negated,swap,
    "\delta_a"]\arrow[r,negated,"t"] & 1
  \end{tikzcd}
\]
that is, by a functor $\autA_T\colon \catIword\to\KlTrans$ or
equivalently, a "$(\KlTrans,1,1)$-automaton". The subsequential
function realized by the transducer $T$ is a partial function
$A^*\rightharpoonup B^*$ and is fully captured by the
"$(\KlTrans,1,1)$-language" $\langL_T\colon \catOword\to\KlTrans$ accepted
by $\autA_T$, which is obtained as $\autA_T\circ\iota$. Indeed, this
$\KlTrans$-language gives for each word $w\in A^*$ a "Kleisli
morphism" $\langL_T(\symbLeft w\symbRight)\colon 1\tokl 1$, or
equivalently, outputs for each word in $A^*$ either a word in $B^*$ or
the undefined element $\bot$.
 
Putting all this together, we can state the following lemma, which
validates the categorical encoding of "subsequential transducers":
\begin{lem}
  "Subsequential transducers" are in one to one correspondence with
  "$(\KlTrans,1,1)$-automata", and partial maps from~$A^*$ to~$B^*$
  are in one to one correspondence with "$(\KlTrans,1,1)$-languages".
  Furthermore, the acceptance of languages is preserved under these
  bijections.
\end{lem}

\begin{rem}
  The morphisms of "$(\KlTrans,1,1)$-automata" as in
  Definition~\ref{def:aut-accepts-lang} are very similar to the
  morphisms of subsequential structures provided
  in~\cite[Definition~3.4]{Hansen10}, with the sole exception that now
  we also have to take into account the "initial states of the
  transducers". On the other hand, Choffrut~\cite{Choffrut2003}
  considers morphisms of transducers that may also output formal
  inverses of words in $\alphabetB^*$. Nevertheless, as discussed
  in~\cite[Remark~3.12]{Hansen10} and as follows from the development
  of the present paper, this is not necessary for minimization.
\end{rem}

In the rest of this section we will see how to obtain Choffrut's
minimization result as an application of Lemma~\ref{lemma:minimal}.
I.e., we have to provide in the category of
"$(\KlTrans,1,1)$-automata",
\begin{enumerate}
\item an "initial object@initial transducer",
\item a "final object@final transducer", and,
\item a "factorization system@factorization transducer".
\end{enumerate}

The existence of the "initial transducer" is addressed in
Section~\ref{subsection:initial transducer}, the one of the "final
transducer" is the subject of Section~\ref{subsection:final
  transducer}.  In Section~\ref{subsection:factorization transducers}
we show how to construct a "factorization system@factorization
transducer".  Putting together all these results, we obtain:

\begin{thm}[Categorical version of 
  \cite{Choffrut79,Choffrut2003}]\label{theorem:minimization
    transducer}
  For every~"$(\KlTrans,1,1)$-language", there exists a minimal
  "$(\KlTrans,1,1)$-automaton" for it.
\end{thm}
Let us note that only the existence of the automaton is mentioned in
this statement, and the way to compute it effectively is not addressed
as opposed to Choffrut's work.  Nevertheless,
Lemma~\ref{lemma:minimal} describes what are the basic functions that
have to be implemented, namely $\Reach$ and $\Obs$.

\AP The rest of this section is devoted to establish the three above
mentioned points.  Unfortunately, as it is usually the case with
"Kleisli categories", $\KlTrans$ is neither complete, nor
cocomplete. It does not even have binary products, let alone countable
powers.  Also, the existence of a non-trivial "factorization system" does not
generally hold in "Kleisli categories". Hence, providing the above
three pieces of information requires a bit of work.

In the next section we present an adjunction between
"$(\KlTrans,1,1)$-automata" and "$(\Set,1,\alphabetB^*+1)$-automata"
which is then used in the subsequent ones for proving the existence of
"initial@initial transducer" and "final automata@final transducer". We
finish the proof with a presentation of the "factorization
system@factorization transducer".

\subsection{Back and forth to automata in set}
\label{subsection:transducer adjunction}

In order to understand what are the properties of the category of
"$(\KlTrans,1,1)$-automata", an important tool will be the ability to
see alternatively a "subsequential transducer" as an "automaton in
$\KlTrans$" as we have seen above, or as an "automaton in $\Set$",
since $\Set$ is much better behaved than $\KlTrans$.  These two points
of view are related through an adjunction, making use of the results
of Section~\ref{sec:lifting-adjucntions-to-automata} and
Lemma~\ref{lem:lifting-fact-syst}.

\AP Indeed, we start from the well known adjunction between $\Set$ and
$\KlTrans$:
\begin{equation}
  \label{eq:adj-kleisli}
  \begin{tikzcd}
    \Set \ar[rr,bend left, "\FreeTrans"] & \bot &\KlTrans\ar[ll,bend
    left, "\UTrans"]\,.
  \end{tikzcd}
\end{equation}
We recall that the free functor $\intro*\FreeTrans$ is defined as the
identity on objects, while for any function $f\colon X\to Y$ the
"morphism@Kleisli morphism" $\FreeTrans f\colon X\tokl Y$ is defined
as $\eta_Y\circ f\colon X\to \monadTrans Y$. For the other direction,
the functor $\intro*\UTrans$ maps an object $X$ in $\KlTrans$ to
$\monadTrans X$ and a morphism $f\colon X\tokl Y$ (which is seen here
as a function $f\colon X\to \monadTrans Y$) to
$\mu_Y\circ \monadTrans f\colon \monadTrans X\to \monadTrans Y$.

\AP A simple, yet important observation is that the language of
interest, which is a partial function $L\colon A^*\rightharpoonup B^*$
can be modeled either as a "$(\KlT, 1,1)$-language" $\langLKlT$, or,
as a "$(\Set, 1,B^*+1)$-language" $\intro*\langLSetTrans$. This is because for
each $w\in A^*$ we can identify $L(w)$ either with an element of
$\KlTrans(1,1)$ or, equivalently, as an element of $\Set(1,B^*+1)$.
\begin{align*}
  \langLKlT\colon\qquad\catOword&\to\KlTrans &     \langLSetTrans\colon\qquad\catOword&\to\Set \\
  \objectLeft&\mapsto 1 &     \objectLeft&\mapsto 1 \\
  \objectRight&\mapsto 1 &     \objectRight&\mapsto B^*+1 \\
  \symbLeft w \symbRight &\mapsto L(w)\colon 1\tokl 1 & \symbLeft w
                                                        \symbRight &\mapsto L(w)\colon 1\to B^*+1
\end{align*}
To see how this fits in the scope of
Section~\ref{sec:lifting-adjucntions-to-automata}, notice that
$\langLKlT$ is an "$(\KlT, \FreeTrans 1,1)$-language", while
$\langLSetTrans$ is an "$(\Set, 1,\UTrans 1)$-language" and they correspond
to each other via the
bijections described in~\eqref{eq:bij-lang}.\\
\AP Applying Lemma~\ref{lem:lifting-adjunctions} for the Kleisli
adjunction~\eqref{eq:adj-kleisli} we obtain an adjunction
$\FreeAut\dashv \AutUTrans$ between the categories of
"$\KlTrans$-automata for $\langLKlT$" and of "$\Set$-automata
accepting $\langLSetTrans$", as depicted in the diagram below.  We will
make heavy use of this correspondence in what follows.
\[
  \begin{tikzcd}
    \catAutoLSet\ar[ddd,"\St"] \ar[rr,bend left=15, "\FreeAut"] & \bot &\catAutoLKlT\ar[ll,bend left=15, "\AutUTrans"]\ar[ddd,"\St"]\\
    & & \\
    & & \\
    \Set \ar[rr,bend left=17, "\FreeTrans"] & \bot &\KlT\ar[ll,bend
    left=17, "\UTrans"]\,.
  \end{tikzcd}
\]

\subsection{The initial \texorpdfstring{$\KlTrans$-automaton}{Kl(T)-automaton} for the language
  \texorpdfstring{$\langLKlT$}{LKl(T)}}
\label{subsection:initial transducer}

\phantomintro{initialSS}\phantomintro{initial transducer}%
The functor $\FreeAut$ is a left adjoint and consequently preserves
colimits and in particular the "initial object". We thus obtain that
the initial $\langL_{\KlT}$-automaton is
$\FreeAut(\AutInit(\langLSetTrans))$, where $\AutInit(\langLSetTrans)$ is the
"initial object" of $\catAutoLSet$.  This automaton can be obtained by
Lemma~\ref{lem:the-minimization-wheel} as the functor
$\AutInit(\langLSetTrans)\colon\catIword\to\Set$ specified by
$\AutInit(\langLSetTrans)(\objectCenter)=A^*$ and for all~$a\in\alphabetA$
\begin{align*}
  \AutInit(\langLSetTrans)(\symbLeft)\colon 1&\to A^* &
                                                   \AutInit(\langLSetTrans)(\symbRight)\colon A^*&\to B^*+1 &
                                                                                                         \AutInit(\langLSetTrans)(a)\colon  A^*&\to A^* \\
  0&\mapsto \varepsilon& w&\mapsto L(w)& w&\mapsto wa
\end{align*}
\AP Hence, by computing the image of $\AutInit(\langLSetTrans)$ under
$\FreeAut$, we obtain the following description of the "initial
$\KlTrans$-automaton@initial automaton" $\AutInit(\langLKlT)$
"accepting" $\langLKlT$: $\AutInit(\langLKlT)(\objectCenter)=A^*$ and
for all $a\in\alphabetA$
\begin{align*}
  \AutInit(\langLKlT)(\symbLeft)\colon 1&\tokl A^* &
                                                     \AutInit(\langLKlT)(\symbRight)\colon A^*&\tokl 1&
                                                                                                        \AutInit(\langLKlT)(a)\colon  A^*&\tokl A^*\\
  0&\mapsto (\varepsilon,\varepsilon)& w&\mapsto L(w)& w&\mapsto
                                                          (\varepsilon,wa)
\end{align*}

\subsection{The final \texorpdfstring{$\KlTrans$-automaton}{Kl(T)-automaton} for the language
  \texorpdfstring{$\langLKlT$}{LKl(T)}}
\label{subsection:final transducer}%

\phantomintro{final transducer}%
The case of the final $\KlTrans$-automaton is more complicated, since
it is not constructed as easily. However, assuming the final automaton
exists, it has to be sent by $\AutUTrans$ to a final
$\Set$-automaton, since $\AutUTrans$ preserves limits. We shall see in
Lemma~\ref{lem:UAut-reflects-final-obj} that $\AutUTrans\colon \catAutoLKlT\to\catAutoLSet$ reflects
  final objects, and hence in order to prove that a
given $\KlTrans$-automaton $\autA$ is a final object of $\catAutoLKlT$
it suffices to show that $\AutUTrans(\autA)$ is the final object in
$\catAutoLSet$. The proof of the following lemma generalizes the fact
that $\UTrans$ reflects final objects and can be proved in the same
spirit.
\begin{lem}
  \label{lem:UAut-reflects-final-obj}
  The functor $\AutUTrans\colon \catAutoLKlT\to\catAutoLSet$ reflects
  final objects.
\end{lem}
\begin{proof}
  \AP Recall that we have the following two adjunctions for the
  categories $\KlT$ of Kleisli algebras, respectively $\intro*\EMT$ of
  Eilenberg-Moore algebras, and the comparison functor
  $K\colon\KlT\to\EMT$ between them.
  \begin{equation}
    \label{eq:4}
    \begin{tikzcd}[column sep={25mm,between origins},row sep={25mm,between origins}]
      \KlT\ar[rr,"K"]\ar[rd, bend left,"\UTrans"] & & \EMT\ar[ld,bend left,"\UTransEM"] \\
      & \Set\ar[lu,bend left,"\FreeTrans"]\ar[ru,bend
      left,swap,"\FreeTransEM"] &
    \end{tikzcd}
  \end{equation}
  \AP The partial function $L\colon A^*\rightharpoonup B^*$ from
  Section~\ref{subsection:transducer adjunction} can also be modeled
  as an $(\EMT,T1,T1)$-language $\intro*\langLEMTTrans\colon\catO\to\EMT$.
  Applying Lemma~\ref{lem:lifting-adjunctions} for the adjunction
  $\intro{\FreeTransEM}\dashv\intro{\UTransEM}$ we obtain an
  adjunction $\FreeAutEM\dashv \AutUTransEM$ between the categories of
  $\EMT$-automata for $\langLEMTTrans$ and of $\Set$-automata for
  $\langLSetTrans$. We also have a lifting
  $ \overline{K}\colon\catAutoLKlT\to\catAutoLEMT $ of the comparison
  functor $K$, which maps a $\KlT$-automaton $\autA$ to the
  $\EMT$-automaton $K\circ \autA$. We obtain the following situation,
  which is just a lifting of diagram~\eqref{eq:4} to the categories of
  automata.
  \[
    \begin{tikzcd}[column sep={25mm,between origins},row
      sep={25mm,between origins}]
      \catAutoLKlT\ar[rr,"\overline{K}"]\ar[rd, bend left,"\AutUTrans"] & & \catAutoLEMT\ar[ld,bend left,"\AutUTransEM"] \\
      & \catAutoLSet\ar[lu,bend left,"\FreeAut"]\ar[ru,bend
      left,swap,"\FreeAutEM"] &
    \end{tikzcd}
  \]
  One can readily check that the functor $\AutUTrans$ is the composite
  $\AutUTransEM\circ \overline{K}$. The functor $\overline{K}$ is full
  and faithful (a property inherited from $K$) and thus reflects final
  objects. On the other hand, the final object in $\catAutoLEMT$ can
  be computed using Lemma~\ref{lem:the-minimization-wheel}, since the
  underlying category $\EMT$ has all limits. Moreover, this final
  automaton is the reflection of the final $\Set$-automaton
  $\AutFin(\langLSetTrans)$.
\end{proof}

We are now ready to describe the final $\KlT$-automaton.  The final
object in $\catAutoLSet$ is the automaton $\AutFin(\langLSetTrans)$ as
described using Lemma~\ref{lem:the-minimization-wheel}.  The functor
$\AutFin(\langLSetTrans)\colon\catI\to\Set$ specified by
\begin{align*}
  &\AutFin(\langLSetTrans)(\objectCenter)=(B^*+1)^{A^*}&
  &\begin{array}{rl}
     \AutFin(\langLSetTrans)(\symbRight)\colon  (B^*+1)^{A^*}&\to B^*+1\\
     K&\mapsto K(\varepsilon)
   \end{array}\\
  &\begin{array}[c]{rl}
     \AutFin(\langLSetTrans)(\symbLeft)\colon 1&\to (B^*+1)^{A^*}\\
     0&\mapsto L
   \end{array}&
  &\begin{array}{rl}
     \AutFin(\langLSetTrans)(a)\colon (B^*+1)^{A^*}&\to (B^*+1)^{A^*}\\
     K&\mapsto \lambda w. K(aw)
   \end{array}
\end{align*}

To describe the set of states of the final automaton in $\catAutoLKlT$
we need to introduce a few notations. Essentially we are looking for a
set of states $Q$ so that $B^*\times Q+1$ is isomorphic to
$(B^*+1)^{A^*}$. The intuitive idea is to decompose each function in
$K\in (B^*+1)^{A^*}$ (\AP except for the one which is nowhere defined,
that is the function $\intro*\botfun=\lambda w.\bot$) into a word in
$B^*$, the common prefix of all the $B^*$-words in the image of $K$,
and an "irreducible function", i.e., a function such that the common
prefix of all the words in the codomain is empty.

\AP For~$v\in B^*$ and a function $K\neq \botfun$ in $(B^*+1)^{A^*}$,
denote by $v\intro*\pstar K$ the function defined for
all~$u\in\alphabetA^*$ by $ (v\pstar K)(u)=v\,K(u)$ if $K(u)\in B^*$
and $(v\pstar K)(u)=\bot$ otherwise.

\AP Define also the ""longest common prefix"" of $K$,
$\intro*\lcpK\in B^*$, as the longest word that is prefix of
all~$K(u)\neq\bot$ for~$u$ in~$A^*$ (this is well defined
since~$K\neq\botfun$).  The ""reduction of~$K$"", $\intro*\red(K)$, is
defined as:
\begin{align*}
  \red(K)(u)&= \begin{cases}
    v&\text{if}~K(u)=\lcp(K)\,v,\\
    \bot&\text{otherwise}.
  \end{cases}
\end{align*}
\AP Finally, $K$ is called ""irreducible"" if $\lcp(K)=\varepsilon$
(or equivalently if~$K=\red(K)$).  We denote by $\intro*\IrrAB$ the
"irreducible" functions in~$(B^*+1)^{A^*}$.

\AP What we have constructed is a bijection $\intro*\bijFinalTrans$ between
\begin{align*}
  \monadTrans(\IrrAB)=B^*\times\IrrAB +1&&\text{and}&&(B^*+1)^{A^*}\ ,
\end{align*}
that is defined as
\begin{align}
  \label{eq:1}
  \begin{array}{rl}
    \bijFinalTrans\colon B^*\times\IrrAB +1 & \to (B^*+1)^{A^*} \\ 
    (u,K)&\mapsto u\pstar K\\
    \bot&\mapsto \botfun\ ,
  \end{array}
\end{align}
and the converse of which maps every $K\neq\botfun$ to
$(\lcpK,\redK)$, and $\botfun$ to~$\bot$.

Given $a\in A$ and $K\in (B^*+1)^{A^*}$ we denote by $a^{-1}K$ the
function in $(B^*+1)^{A^*}$ that maps $w\in A^*$ to $K(aw)$.

We can now define the "automaton"
$\AutFin(\langLKlT)\colon\catI\to \KlTrans$ by setting
\begin{itemize}
\item $\AutFin(\langLKlT)(\objectLeft)=1$,
\item $\AutFin(\langLKlT)(\objectCenter)=\IrrAB$, and,
\item $\AutFin(\langLKlT)(\objectRight)=1$, 
\end{itemize}
and defining $\AutFin(\langLKlT)$ on arrows as follows
\begin{align*}
    & \AutFin(\langLKlT)(\symbLeft)\colon 1\tokl \IrrAB  & & 0\mapsto (\lcpL, \redL)\\
    & \AutFin(\langLKlT)(\symbRight)\colon  \IrrAB \tokl 1 & & K\mapsto K(\varepsilon)\\
    & \AutFin(\langLKlT)(a)\colon \IrrAB \tokl \IrrAB & & K\mapsto
                                                          \begin{cases}
                                                            \bot &  \text{ if } a^{-1}K=\botfun\,, \\
                                                            (\lcp(a^{-1}K),\red(a^{-1}K))
                                                            & \text{
                                                              otherwise}
                                                          \end{cases}
\end{align*}

\begin{lem}
  The $\KlTrans$-automaton $\AutFin(\langLKlT)$ is a final object in
  $\catAutoLKlT$.
\end{lem}

\begin{proof}
  We show that $\AutUTrans(\AutFin(\langLKlT))$ is isomorphic to the
  final automaton $\AutFin(\langLSetTrans)$. Indeed, at the level of
  objects the bijection between the sets
  $\AutUTrans(\AutFin(\langLKlT))(\objectCenter)$ and
  $\AutFin(\langLSetTrans)(\objectCenter)$ is given by the function
  $\bijFinalTrans$ defined in~\eqref{eq:1}. It is easy to check that also on
  arrows $\AutUTrans(\AutFin(\langLKlT))$ is the same as
  $\AutFin(\langLSetTrans)$ up to the correspondence given by the function
  $\bijFinalTrans$.
\end{proof}

\subsection{A factorization system on \texorpdfstring{$\catAutoLKlT$}{Auto(LKl(T))}}
\label{subsection:factorization transducers}

\intro[factorization transducer]{}%
The factorization system on $\catAutoLKlT$ is obtained using
Lemma~\ref{lem:lifting-fact-syst} from a factorization system on
$\KlTrans$. There are several non-trivial factorization systems on
$\KlTrans$, one of which is obtained from the regular epi-mono
factorization system on $\Set$, or equivalently, from the regular
epi-mono factorization system on the category of Eilenberg-Moore
algebras for~$\monadTrans$. Notice that this is a specific result for
the monad $\monadTrans$ since in general, there is no reason that the
Eilenberg-Moore algebra obtained by factorizing a morphism between
free algebras be free itself. Nevertheless, in order to capture
precisely the syntactic transducer defined by
Choffrut~\cite{Choffrut79,Choffrut2003}, we will provide yet another
factorization system $(\EpiKlT,\MonoKlT)$, which we define concretely
as follows.  Given a morphism $f\colon X\tokl Y$ in $\KlTrans$ we
write $\pi_1(f)\colon X\to B^*+\{\bot\}$ and
$\pi_2(f)\colon X\to Y+\{\bot\}$ for the `projections' of $f$, defined
by
\begin{align*}
  \pi_1(f)(x)&=
               \begin{cases}
                 u & \text{if }f(x)=(u,y)\,,\\
                 \bot & \text{otherwise,}
               \end{cases}
                   &
                     \text{and}\qquad 
                     \pi_2(f)(x)&=
                                  \begin{cases}
                                    y & \text{if }f(x)=(u,y)\,,\\
                                    \bot & \text{otherwise.}
                                  \end{cases}
\end{align*}
We say that a partial function $g\colon X\to Y+\{\bot\}$ is surjective
when for every $y\in Y$ there exists $x\in X$ so that $g(x)=y$.

\AP The class $\intro{\EpiKlT}$ consists of all the morphisms of the
form $e\colon X\tokl Y$ such that $\pi_2(e)$ is surjective and the
class $\intro{\MonoKlT}$ consists of all the morphisms of the form
$m\colon X\tokl Y$ such that $\pi_2(m)$ is injective and $\pi_1(m)$ is
the constant function mapping every $x\in X$ to $\varepsilon$.

\begin{lem}
  $(\EpiKlT, \MonoKlT)$ is a factorization system on $\KlTrans$.
\end{lem}

\begin{proof}
  Notice that $f$ is an isomorphism in $\KlTrans$ if and only if
  $f\in\EpiKlT\cap\MonoKlT$.

  If $f\colon X\tokl Y$ is a morphism in $\KlTrans$ then we can define
  \[Z=\{y\in Y\mid \exists x\in X\ldotp\exists u\in B^*\ldotp
    f(x)=(u,y)\}\,.\] We define $e\colon X\tokl Z$ by $e(x)=f(x)$ and
  $m\colon Z\tokl Y$ by $m(y)=(\varepsilon,y)$. One can easily check that
  $f=m\circ e$ in $\KlTrans$.

  Lastly, we can show that the "diagonal property" holds. Assume we
  have a commuting square in $\KlTrans$.
  \[
    \begin{tikzcd}
      X\ar[r, two heads, "e"]\ar[d,swap, "f"] & Y\ar[d, "g"]\ar[dl, dashed,"d"] \\
      Z\ar[r, tail,swap, "m"] & W
    \end{tikzcd}
  \]
  We will prove the existence of $d\colon Y\tokl Z$ so that
  $d\circ e=f$ and $m\circ d=g$. Assume $y\in Y$. If $g(y)=\bot$ we
  set $d(y)=\bot$. Otherwise assume $g(y)=(v,t)$, for some $v\in B^*$
  and $t\in W$. Since $e\in \EpiKlT$, there exists $u\in B^*$ and
  $x\in X$ so that $e(x)=(u,y)$. Assume $f(x)=(w,z)$ for some
  $w\in B^*$ and $z\in Z$. We set $d(y)=(v, z)$. First, we have to
  prove that this definition does not depend on the choice of $x$.

  Assume that we have another $x'\in X$ so that $e(x')=(u',y)$ and
  assume $f(x')=(w',z')$. Using the fact that $m\in \MonoKlT$, we will
  show that $z=z'$, and thus $d(y)$ is well defined. Indeed, notice
  that

\begin{tabular}{lll}
  $\begin{cases}
    g\circ e(x)=(uv,t)\\
    g\circ e(x')=(u'v,t)\,,\\
  \end{cases}
  $
  &
    or equivalently, 
  &
    $
    \begin{cases}
      m\circ f(x)=(uv,t)\\
      m\circ f(x')=(u'v,t)\,,\\
    \end{cases}
  $
\end{tabular}

\noindent Assume that $m(z)=(\varepsilon, t_1)$ and
$m(z')=(\varepsilon, t_2)$. This entails
\[
  \begin{cases}
    m\circ f(x)=(uv,t)=(w,t_1)\\
    m\circ f(x')=(u'v,t)=(w',t_2)\,.\\
  \end{cases}
\]
We obtain that $t_1=t_2=t$. Since $m\in\MonoKlT$ (and thus $\pi_2(m)$
is injective) we get that $z=z'$, which is what we wanted to prove.
It is easy to verify that $d\circ e=f$ and $m\circ d=g$.
\end{proof}

This completes the proof of Theorem~\ref{theorem:minimization
  transducer}.

\section{Brzozowski's minimization algorithm}
\label{sec:brzozowski}

\subsection{Presentation}

""Brzozowski's algorithm"" is a minimization algorithm for
automata. It takes as input a "non-deterministic
automaton"~$\mathcal{A}$, and computes the deterministic automaton:
\begin{align*}
  \determinize(\transpose(\determinize(\transpose(\mathcal A)))),
\end{align*}
in which
\begin{itemize}
\item\AP $\intro*\determinize$ is the operation from classical
  automata theory that takes as input a deterministic automaton,
  applies a powerset construction and at the same time restricts to
  the reachable states, yielding a deterministic automaton, and
\item\AP $\intro*\transpose$ is the operation that takes as input an
  non-deterministic automaton reverses all its edges, and swaps the
  role of initial and final states (it accepts the mirrored language).
\end{itemize}
\AP In this section, we will establish the correctness of Brzozowski's
algorithm: this sequence of operations yields the minimal automaton
for the language. For easing the presentation we shall present the
algorithm in the form:
\begin{align*}
  \determinize(\codeterminize(\mathcal A)),
\end{align*}
in which $\intro*\codeterminize$ is the operation that takes a
non-deterministic automaton, and constructs a backward deterministic
one (it is equivalent to the sequence
$\transpose\circ\determinize\circ\transpose$).

In the next section, we show how $\determinize$ and~$\codeterminize$
can be seen as adjunctions, and we use it immediately after in a
correctness proof of "Brzozowski's algorithm".

We will use the representation of "non-deterministic automata" as
$(\Rel,1,1)$-automata (see Example~\ref{ex:automata-as-functors}) and
the fact that determinization, respectively codeterminization, can be
seen as right, respectively left, adjoints, as discussed in
Section~\ref{sec:det-codet-adj}.


\subsection{Brzozowski's minimization algorithm}
\label{sec:duality}

The correctness of "Brzozowski's algorithm" can be seen in the
following chain of adjunctions from Lemma~\ref{lem:adj-reach-obs} and
diagram~\eqref{eq:2} (that all correspond to equivalences at the level of
languages):
\[
  \begin{tikzcd}[column sep={14mm,between origins}]
    \catReachAutoLSet \ar[rr,bend left, hook,"E"] &\bot
    &\catAutoLSet\ar[rr,bend left, "\overline{\FPow}"]\ar[ll,bend
    left, "\Reach"] & \bot &\catAutoLRel\ar[ll,bend left,
    "\overline{\UPow}"] \ar[rr,bend left, "\overline{\UPowop}"] & \bot
    & \catAutoLSetop \ar[ll,bend left,
    "\overline{\FPowop}"]\ar[rr,bend left, "\Obs"] &\bot~ &
    \catObsAutoLSetop\ar[ll,bend left, hook,"E"]
  \end{tikzcd}
\]

A path in this diagram corresponds to a sequence of transformations of
automata. It happens that when~$\Obs$ is taken, the resulting
automaton is "observable", i.e., there is an injection from it to the
"final object". This property is preserved under the sequence of right
adjoints $\Reach\circ\overline{\UPow}\circ 
\overline{\FPowop}\circ E$. Furthermore, after application
of~$\Reach$, the automaton is also "reachable". This means that
applying the sequence
$\Reach\circ\overline{\UPow}\circ\overline{\FPowop}\circ
E\circ\Obs\circ\overline{\UPowop}$ 
to a non-deterministic automaton produces a deterministic and minimal
one for the same language.  We check for concluding that the sequence
$\Obs\circ\overline{\UPowop}$ is what is implemented by
$\codeterminize$, that the composite $\overline{\FPowop}\circ
E$ essentially transforms a backward deterministic "observable"
automaton into a non-deterministic one, and that finally
$\Reach\circ\overline{\UPow}$ is what is implemented by
$\determinize$. Hence, this indeed is "Brzozowski's algorithm".

\begin{rem}
  The composite of the two adjunctions in~\eqref{eq:2} is almost the
  adjunction of~\cite[Corollary~9.2]{BonchiBHPRS14} upon noticing that
  the category $\catAutoLSetop$ of
  $\Set^\op$-automata accepting a language
  $\langLSetop$ is isomorphic to the opposite of the category
  $\catAutoLSetrev$ of
  $\Set$-automata that accept the reversed language seen as functor
  $\langLSetop$. This observation in turn can be proved using the
  symmetry of the input category $\catI$.
\end{rem}

\subsection{Weighted Brzozowski's minimization algorithm}
\label{sec:w-Brz}

The weighted version of Brzozowski's minimization algorithm presented
in~\cite{BonchiBHPRS14} can also be explained in our framework using
the chain of adjunctions described in the next diagram. Given a
semiring $S$, a weighted language $L\colon A^*\to S$, can be modeled
either as a functor $\langL_{\SMod}\colon\catOword\to\SMod$ or,
equivalently, as a functor
$\langL_{\SMod^{\mathit{op}}}\colon\catOword\to\SMod^{\mathit{op}}$,
using the observation that
$\SMod(S,S)\cong\SMod^\mathit{op}(S,S)\cong S$. Notice that the
category $\catAutoLSModop$ is the opposite of the category of automata
accepting the reversed language $L^\mathit{rev}$, defined by
$L^\mathit{rev}(w)=L(w^\mathit{rev})$. The adjunction between $\SMod$
and its opposite obtained by taking the dual spaces, lifts by virtue
of Lemma~\ref{lem:lifting-adjunctions} to an adjunction between the
corresponding categories of automata, where the lifting
$\overline{S^{-}}$ is essentially reversing the automaton. Thus the
operations of the weighted Brzozowski's algorithm correspond to a path
$E\circ\Reach\circ\overline{S^{-}}\circ E\circ \Obs\circ
\overline{S^{-}}$ in the next diagram.
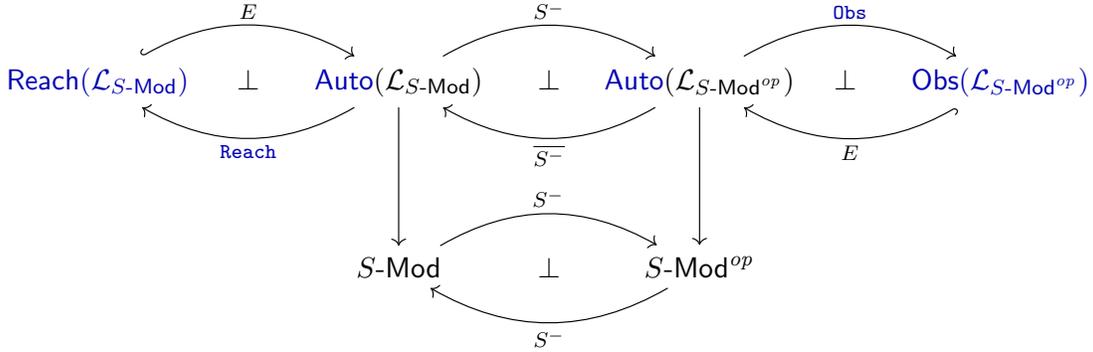
\[
  \begin{tikzcd}[column sep={20mm,between origins},row sep={25mm,between origins}]
    \catReachAutoLSMod \ar[rr,bend left, hook,"E"]
    &
    \bot
    &
    \catAutoLSMod\ar[rr,bend left, "\overline{S^{-}}"]\ar[ll,bend
    left, "\Reach"]\ar[d]
    &
    \bot
    &
    \catAutoLSModop \ar[ll,bend left,
    "\overline{S^{-}}"]\ar[rr,bend left, "\Obs"]\ar[d]
    &\bot~
    &
    \catObsAutoLSModop\ar[ll,bend left, hook,"E"]\\
    &
    &
    \SMod\ar[rr,bend left, "S^{-}"]
    &
    \bot
    &
    \SMod^{\mathit{op}}\ar[ll,bend left,
    "S^{-}"]
      &
      &
  \end{tikzcd}
\]

\section{Monoids for language recognition}
\label{sec:syntactic-monoids}

In this section we show that the notion of a syntactic monoid for a
given language $L\subseteq A^*$ fits in the functorial framework
introduced in this paper. We argue that the syntactic monoid can be
obtained using the generic principles outlined in
Section~\ref{sec:lang-auto-functors} by changing accordingly the input
category.  However, it is not the monoids recognizing a language that
will be modeled as $\Set$-valued functors, but rather "biaction
recognizers". We will prove that we have initial and final biactions
recognizing a language and that the minimal biaction recognizing a
language (obtained via an epi-mono factorization) can be in fact
equipped with a monoid structure and yields precisely the syntactic
monoid for that language.

We start with the definitions of monoid and biaction recognizers.

\begin{defi}
  \AP We call ""monoid recognizer"" a tuple
  $(\phi\colon A^*\to M, P\subseteq M)$ consisting of a monoid
  homomorphism $\phi$ and a subset $P$ of $M$. It recognizes the
  language $\{w\in A^*\mid \phi(w)\in P\}$. We say that a "monoid
  recognizer" $(\phi\colon A^*\to M, P\subseteq M)$ is a ""surjective
  monoid recognizer"" when the morphism $\phi$ is surjective. A
  morphism between "monoid recognizers"
  $(\phi\colon A^*\to M, P\subseteq M)$ and
  $(\phi'\colon A^*\to M', P'\subseteq M')$ is a monoid morphism
  $h\colon M\to M'$ such that $h\circ \phi=\phi'$ and $h^{-1}(P')=P$.
\end{defi}

\begin{defi}
  \AP An ""$A^*$-biaction"" is a set $X$ equipped with left and right
  $A^*$-actions\footnote{both denoted by $\cdot$ by abuse of notation}
  $\cdot\colon A^*\times X\to X$ and $\cdot\colon X\times A^*\to X$
  which commute, that is, $(u\cdot x)\cdot v=u\cdot(x\cdot v)$ for all
  $u,v\in A^*$ and $x\in X$. Morphisms of $A^*$-biactions are
  functions that are morphisms of both the left and right
  $A^*$-actions.
\end{defi}

\begin{defi}
  \AP An ""$A^*$-biaction recognizer"" is a tuple
  $(\phi\colon A^*\to X, P\subseteq X)$ where $\phi$ is a morphism of
  "$A^*$-biactions" and $P$ is a subset of $X$. We call the elements
  of $P$, the ""accepting@acceptingelements"" elements of $X$. The
  language recognized by $(\phi\colon A^*\to X, P\subseteq X)$ is the
  set $\{w\in A^*\mid \phi(w)\in P\}$. The "$A^*$-biaction recognizer"
  is called ""surjective@surjective $A^*$-biaction recognizer"" when
  $\phi$ is so.  A morphism between "$A^*$-biaction recognizers"
  $(\phi\colon A^*\to X, P\subseteq M)$ and
  $(\phi'\colon A^*\to X', P'\subseteq M')$ is an $A^*$-biaction
  morphism $h\colon X\to X'$ such that $h\circ \phi=\phi'$ and
  $h^{-1}(P')=P$.
\end{defi}

The proof of the next lemma is straightforward. The second part was
used in~\cite{GehrkePR16}.

\begin{lem}\label{lem:biactions-2-monoids}
  Any "monoid recognizer" is an "$A^*$-biaction
  recognizer". Conversely, any "surjective $A^*$-biaction recognizer"
  is a "surjective monoid recognizer".
\end{lem}

\AP In order to describe $A^*$-biactions as functors, we will consider
the category $\intro{\catIMon}$ described in the diagram below. Just
as the input category for word automata, $\catIMon$ has three objects:
$\objectLeft$, $\objectCenter$ and $\objectRight$.  The homsets in
this category can be intuitively described as follows:
\begin{itemize}
\item $\catIMon(\objectLeft,\objectCenter)$ and
  $\catIMon(\objectLeft,\objectRight)$ are both isomorphic to the set
  of finite words over $A$;
\item the sets $\catIMon(\objectCenter,\objectCenter)$ and
  $\catIMon(\objectCenter,\objectRight)$ consist of the finite words
  over $A$ with one ``hole''.
\end{itemize}
\AP Concretely, for each $w\in A^*$ we have morphisms
$\intro*\wmorph{w}\colon\objectLeft\to\objectCenter$ and
$\intro*\wfinmorph{w}\colon\objectLeft\to\objectRight$.  For every two
words $u,v\in A^*$ we have morphisms
$\intro*\uvmorph{u}{v}\colon\objectCenter\to\objectCenter$ and
$\intro*\uvfinmorph{u}{v}\colon\objectCenter\to\objectRight$.
\begin{equation}\label{eq:input-cat-mon}
  \begin{tikzcd}
    \objectLeft\arrow[r,"\wmorph{w}"]\arrow[rr,bend
    right,"\wfinmorph{w'}"] &
    \objectCenter\arrow[loop,looseness=6,swap,
    "\uvmorph{u}{v}"]\arrow[r,"\uvfinmorph{u'}{v'}"] & \objectRight
  \end{tikzcd}
\end{equation}
The composition defined as a substitution of the $\msquare$ symbol. We
define it formally as follows:
\begin{align*}
  \uvmorph{u}{v}\circ\wmorph{w}&=\wmorph{uwv}\colon\objectLeft\to\objectCenter\ ,&  
                                                                                   \uvmorph{u'}{v'}\circ\uvmorph{u}{v}&=\uvmorph{u'u}{vv'}\colon\objectCenter\to\objectCenter\ ,\\
  \uvfinmorph{u}{v}\circ\wmorph{w}&=\wfinmorph{uwv}\colon\objectLeft\to\objectRight\ ,&
                                                                                        \text{and}\qquad
                                                                                        \uvfinmorph{u'}{v'}\circ\uvmorph{u}{v}&=\uvfinmorph{u'u}{vv'}\colon\objectCenter\to\objectRight
\end{align*}

\begin{defi}
  \AP An ""$(\catIMon,\Set,1,2)$-automaton"" is a functor
  $\biactAutA\colon\catIMon\to\Set$ such that
  $\biactAutA(\objectLeft)=1$ and $\biactAutA(\objectRight)=2$. A
  morphism of "$(\catIMon,\Set,1,2)$-automata" is a natural
  transformation $\alpha\colon\biactAutA\to\biactAutB$ so that both
  $\alpha_{\objectLeft}$ and $\alpha_{\objectRight}$ are the identity
  morphisms on $1$, respectively~$2$.
\end{defi}

\AP Let $\intro{\catOMon}$ be the full subcategory of $\catIMon$ on
objects $\objectLeft$ and $\objectRight$. We denote by $\iota$ the
inclusion
\[
  \begin{tikzcd}
    \catOMon\arrow[r,hook,"\iota"] & \catIMon\,.
  \end{tikzcd}
\]
The language accepted by an "$(\catIMon,\Set,1,2)$-automaton"
$\biactAutA\colon\catIMon\to\Set$ is the composite functor
$\langL=\biactAutA\circ\iota\colon\catOMon\to\Set$. Just as for word
automata, the functor $\langL$ encodes a language $L\subseteq A^*$
consisting of the words $w\in A^*$ such that the map
$\langL(\wfinmorph{w})\colon 1\to 2$ is constant to $1$, that is,
$\langL(\wfinmorph{w})(0)=1$.

\begin{lem}\label{lem:biact-functors}
  The category of "$(\catIMon,\Set,1,2)$-automata" is equivalent to
  that of "$A^*$-biaction recognizers".
\end{lem}
\begin{proof}
  Consider an "$(\catIMon,\Set,1,2)$-automaton"
  $\biactAutA\colon\catIMon\to\Set$. Then, the set
  $Q=\biactAutA(\objectCenter)$ can be equipped with commuting left
  and right $A^*$-actions. Indeed, we define the left action
  $\cdot\colon A^*\times Q\to Q$ by
  $u\cdot q=\biactAutA(\uvmorph{u}{\varepsilon})(q)$ and the right
  action $\cdot\colon Q\times A^*\to Q$ by
  $q\cdot v=\biactAutA(\uvmorph{\varepsilon}{v})(q)$.  We define
  $\phi\colon A^*\to Q$ by $\phi(w)=\biactAutA(\wmorph{w})(0)$.

  One can check that $\biactAutA$ being a functor entails that these
  are well defined and commuting left and right actions and that
  $\phi$ is a morphism of biactions.  Furthermore, we define the
  subset of "accepting@acceptingelements" elements of $Q$ as the
  subset whose characteristic function is the morphism
  $\biactAutA(\uvfinmorph{\varepsilon}{\varepsilon})\colon Q\to 2$.

  Conversely, given an "$A^*$-biaction recognizer"
  $(\phi\colon A^*\to X, P\subseteq X)$ we define
  $\biactAutA\colon\catIMon\to\Set$ as follows. We define
  $\biactAutA(\objectCenter)=X$ and we put
  \begin{itemize}
  \item $\biactAutA(\wmorph{w})(0)=\phi(w)$;
  \item $\biactAutA(\uvmorph{u}{v})(x)=\phi((u\cdot x)\cdot v)$;
  \item $\biactAutA(\uvfinmorph{u}{v})(x)=\chi_P((u\cdot x)\cdot v)$;
  \end{itemize}
  where $\chi_P$ is the characteristic function of $P$.
\end{proof}

\begin{rem}
  The functors constructed in the proof of
  Lemma~\ref{lem:biact-functors} preserve the accepted language, that
  is, for each language $L\subseteq A^*$ we obtain an equivalence
  between the categories of "$(\catIMon,\Set,1,2)$-automata" accepting
  $L$, respectively of $A^*$-biaction recognizers for $L$.
\end{rem}

The following lemma immediately follows from
Corollary~\ref{cor:exist-Kan-suff-cond}.

\begin{lem}
  Given a language $\langL\colon\catOMon\to\Set$, the initial and
  final "$(\catIMon,\Set,1,2)$-automata" accepting $\langL$ exist and
  can be computed as left, respectively right Kan extensions of
  $\langL$ along $\iota\colon\catOMon\to\catIMon$. 
  
  Furthermore, the minimal "$(\catIMon,\Set,1,2)$-automaton" accepting $\langL$ is
  obtained via the factorization \[
    \begin{tikzcd}
      \Lan{\langL}{\iota}\ar[r,two heads] & \Min(\langL)\ar[r,tail] &
      \Ran{\langL}{\iota}\,.
    \end{tikzcd}
  \]
\end{lem}

Using the colimit computation of the left Kan extension
$\Lan{\langL}{\iota}$, we obtain the concrete description of the
initial automaton $\autAinitL$ accepting $\langL$. We have that
\[
  \autAinitL(\objectCenter)=\coprod\limits_{u\in A^*} 1 \simeq A^*
\]
and for all $w,u,v\in A^*$ we have
\begin{itemize}
\item $\autAinitL(\wmorph{w})\colon 1\to A^*$,
  $0\in 1\mapsto w\in A^*$;
\item $\autAinitL(\uvmorph{u}{v})\colon A^*\to A^*$, $w\mapsto uwv$;
\item $\autAinitL(\uvfinmorph{u}{v})\colon A^*\to 2$,
  $w\mapsto \langL(\wfinmorph{uwv})(0)$.
\end{itemize}

Dually, using the limit computation of the right Kan extension
$\Ran{\langL}{\iota}$, we obtain the concrete description of the final
automaton $\autAfinalL$ accepting $\langL$. We have that
\[
  \autAfinalL(\objectCenter)=\prod\limits_{u,v\in A^*} 2 \simeq
  2^{A^*\times A^*}
\]
and for all $w,u,v\in A^*$ we have
\begin{itemize}
\item $\autAfinalL(\wmorph{w})\colon 1\to 2^{A^*\times A^*}$,
  $0\in 1\mapsto \{(u,v)\in A^*\times A^*\mid
  \langL(\wfinmorph{uwv})(0)=1 \}$;
\item
  $\autAfinalL(\uvmorph{u}{v})\colon 2^{A^*\times A^*}\to 2^{A^*\times
    A^*}$, $B\mapsto \{(u',v')\in A^*\times A^*\mid (uu',v'v)\in B\}$;
\item $\autAfinalL(\uvfinmorph{u}{v})\colon 2^{A^*\times A^*}\to 2$,
  $B\mapsto
  \begin{cases}
    1, & (u,v)\in B\\
    0, & \text{ otherwise }
  \end{cases}
  $
\end{itemize}

We thus obtain a diagram similar to that for word automata below
Lemma~\ref{lem:the-minimization-wheel}.

\begin{equation}
  \label{eq:synt-monoid-set}
  \begin{tikzcd}[column sep={2cm,between origins},row
    sep={1.7cm,between origins}]
    & \coprod\limits_{u\in A^*} 1
    \arrow[rd,bend left, "\langL(u-v)?"]\arrow[d,two heads] &
    \\
    1
    \arrow[bend right,swap,"\langL(-w-)"]{rd} \arrow[ru,,bend left,
    "w"] \arrow[r,"i"] & \Min(\langL) \arrow[r,"f"] \arrow[d, tail] &
    2
    \\
    & \prod\limits_{(u,v)\in A^*\times A^*}2 \arrow[bend
    right]{ru}[swap]{(u,v)?} &
  \end{tikzcd}
\end{equation}

The unique natural transformation $\alpha$ from the initial automaton
$\autAinitL$ to the final automaton $\autAfinalL$ is determined by the
function $\alpha_{\objectCenter}\colon A^*\to 2^{A^*\times A^*}$
defined by
\[
  w\in A^*\mapsto\{(u,v)\in A^*\times A^*\mid uwv\in L\}
\]
The epi-mono factorization of this maps yields precisely the quotient
of $A^*$ by the syntactic congruence
\[
  w\sim_L w' \text{ if and only if } \forall u,v\in A^*\ uwv\in
  L\Leftrightarrow uw'v\in L\,,
\]
that is, the carrier set of the syntactic monoid for the language $L$.

\begin{thm}
  \AP The minimal automaton $\Min(\langL)$ corresponds to the
  syntactic monoid $\intro\Syn(L)$ of the language $L$.
\end{thm}
\begin{proof}
  The minimal "$(\catIMon,\Set,1,2)$-automaton" accepting $\langL$
  corresponds via the equivalence of Lemma~\ref{lem:biact-functors} to
  a "surjective $A^*$-biaction recognizer". Thus, using
  Lemma~\ref{lem:biactions-2-monoids}, we obtain that $\Min(\langL)$
  corresponds to a "surjective monoid recognizer" for $L$, $\Syn(L)$,
  with carrier set $\Min(\langL)(\objectCenter)$. Consider any other
  "monoid recognizer" $(\phi\colon A^*\to M, P\subseteq M)$ for
  $L$. By the first part of Lemma~\ref{lem:biactions-2-monoids}, any
  "monoid recognizer" can be seen as an "$A^*$-biaction recognizer",
  and thus as an "$(\catIMon,\Set,1,2)$-automaton" $\biactAutA$. By
  Lemma~\ref{lemma:minimal}, we know that $\Min(\langL)$ is isomorphic
  to $\Obs(\Reach(\biactAutA))$, that is, to a quotient of a
  sub-automaton of $\biactAutA$. One can easily check that
  $\Reach(\biactAutA)$ corresponds to a "surjective $A^*$-biaction
  recognizer", and thus to a "surjective monoid recognizer" for
  $L$. If follows that $M$ is the quotient of a submonoid of
  $\Syn(L)$.
\end{proof}

\section{Conclusion}
\label{sec:conclusion}
In this paper we propose a view of automata as functors and we showed
how to recast well understood classical constructions in this setting,
and in particular minimization of subsequential transducers.  We argue
that this perspective gives a unified view of language recognition and
syntactic objects.

In a similar vein to the developments of
Section~\ref{sec:syntactic-monoids}, we can obtain the syntactic
algebras recognizing languages for any algebraic theory over the
category $\Set$. The category $\catIMon$ is specific to the algebraic
theory of monoids: one can notice that the hom-sets
$\catIMon(\objectCenter,\objectRight)$, respectively
$\catIMon(\objectCenter,\objectRight)$ are isomorphic to ``contexts
(or terms) with one hole'' also known as linear unary polynomials in
universal algebra. In order to obtain a similar treatment for
syntactic algebras for an arbitrary algebraic theory, one should
change the input category and replace $\catIMon$ by a similar
category, but where the morphisms are either terms or linear unary
polynomials. A simpler input category could be obtained by considering
the notion of ``unary presentation'' developed
in~\cite{UrbatACM:Eilenberg17}.

We can go beyond regular languages and obtain in this fashion the
``syntactic space with an internal monoid'' of a possibly non-regular
language~\cite{GehrkePR16}. To this end one would just have to compute
the product and the coproduct in~\eqref{eq:synt-monoid-set} in the
category of Stone spaces.

We hope we can extend the framework to work with tree automata in
monoidal categories. We discussed mostly NFA determinization, but we
can obtain a variation of the generalized powerset
construction~\cite{SilvaEtAl:genPow} in this framework.

\bibliographystyle{plain} \bibliography{refs}

\end{document}